\documentclass[article]{article}
\usepackage{color}
\usepackage{graphicx}
\usepackage{amsmath}
\usepackage{amssymb}
\usepackage{mathrsfs}
\usepackage[numbers,sort&compress]{natbib}
\usepackage{amsthm}
\numberwithin{equation}{section}
\usepackage{pgf}
\usepackage{CJK}
\usepackage{graphicx}
\usepackage{tikz}
\usepackage{stmaryrd}
\usepackage{graphicx}
\usepackage{hyperref}
\usepackage[latin1]{inputenc}
\usepackage[T1]{fontenc}
\usepackage[english]{babel}
\usepackage{listings}
\usepackage{xcolor,mathrsfs,url}
\usepackage{amssymb}
\usepackage{amsmath}
\usepackage{ifthen}
\newtheorem{theorem}{Theorem}
\newtheorem{lemma}{Lemma}
\newtheorem{corollary}{Corollary}
\newtheorem{proposition}{Proposition}
\newtheorem{Assumption}{Assumption}
\usepackage {caption}
\usepackage{float}
\usepackage{subfigure}

\DeclareMathOperator*{\res}{Res}

\begin{document}

\title{ Soliton Resolution for the 	Short-pluse  Equation }
\author{Yiling YANG$^1$ and Engui FAN$^{1}$\thanks{\ Corresponding author and email address: faneg@fudan.edu.cn } }
\footnotetext[1]{ \  School of Mathematical Sciences, Fudan University, Shanghai 200433, P.R. China.}

\date{ }

\maketitle
\begin{abstract}
	\baselineskip=17pt

In this paper, we  apply $\overline\partial$ steepest descent method  to study the Cauchy problem for the focusing nonlinear  short-pluse   equation
\begin{align}
&u_{xt}=u+\frac{1}{6}(u^3)_{xx},  \nonumber\\
&u(x,0)=u_0(x)\in H^{1,1}(\mathbb{R}),\nonumber
\end{align}	
 where   $H^{1,1}(\mathbb{R})$ is  a weighted Sobolev space.
We construct the solution of the short-pluse equation  via a the solution of Riemann-Hilbert problem in the new scale $(y,t)$.
In any fixed space-time cone of the new scale $(y,t)$ which stratify that $v_1\leq v_1 \in \mathbb{R}^-$ and $\xi=\frac{y}{t}<0$,
\begin{equation}
C(y_1,y_2,v_1,v_2)=\left\lbrace (y,t)\in R^2|y=y_0+vt, y_0\in[y_1,y_2]\text{, }v\in[v_1,v_2]\right\rbrace,\nonumber
\end{equation}  we   compute the long time asymptotic expansion of the solution $u(x,t)$,
which prove soliton resolution conjecture consisting of three terms:  the leading order term can
be characterized with  an $N(I)$-soliton whose parameters are modulated by
a sum of localized soliton-soliton
 interactions as one moves through the cone; the second $t^{-1/2}$ order term  coming from soliton-radiation interactions on continuous   spectrum  up to an residual error order
 $\mathcal{O}(|t|^{-1})$ from a $\overline\partial$ equation. Our results also show that soliton solutions of
  short-pluse  equation are asymptotically stable.

\end{abstract}

\baselineskip=17pt

\newpage

\tableofcontents

\section {Introduction}
\quad In this paper, we study the long time asymptotic behavior for   the initial value problem  of the short
pulse (SP) equation
\begin{align}
&u_{xt}=u+\frac{1}{6}(u^3)_{xx},\label{sp}\\
&  u(x,0)=u_0(x),\label{sp2}
\end{align}
where $u(x,t)$ is a real-valued function, which represents the magnitude of the electric field,  and the   initial data  $u_0(x)$  belongs to  the weighted Sobolev space
\begin{equation*}
H^{1,1}(R)=\left\lbrace f\in L^2(R);xf,f'\in L^2(R)\right\rbrace.
\end{equation*}
The SP equation (\ref{sp})  was proposed  to describe the propagation of ultra-short optical pulses in silica optical fibers \cite{sp1}. Pulse propagation in optical fibers is usually modeled by the cubic nonlinear Schr$\ddot{o}$dinger (NLS) equation \cite{GP}.  So  NLS equation forms the basis for optimizing existing fiber links and suggesting new fiber communication systems in attempts to achieve high bit-rate data transmission. However, it is questionable that  the validity of the NLS equation  as a slowly varying amplitude approximation of Maxwell's equations can  describe the propagation of these very narrow pulseless. But the SP equation describes the evolution of a short pulse in nonlinear media if the pulse center is far from the nearest resonance frequency of the material's susceptibility. In some sense it represents the opposite extreme from the NLS approximation since that results from expanding the susceptibility in the frequency while equation results from expanding the susceptibility in the wavelength, which means  the SP equation provides an increasingly better approximation to the corresponding solution of the Maxwell equations as the pulse duration shortens \cite{Y2005}. However, we can find that the SP equation appeared first as one of Rabelo's equations which describe photospherical surfaces, possessing a zero-curvature representation \cite{1989ML}.

Some scholar found that the SP equation can be viewed as the short wave approximation to the  modified Camassa-Holm equation \cite{BFch1,BFch2,PJO,Qiao}
\begin{equation}
m_t+\left((u^2-u^2_x)m \right) _x+u_x=0,\hspace{0.5cm}m=u-u_{xx}.\label{CH}
\end{equation}
 Actually,  by introducing the new variables
\begin{equation}
x'=\frac{x}{\epsilon},\hspace{0.5cm}t'=t\epsilon,\hspace{0.5cm}u'=\frac{u}{\epsilon^2},\nonumber
\end{equation}
passing to the limit $\epsilon\to 0 $ and retaining the main terms, we can  reduce the modified Camassa-Holm  equation (\ref{CH})    to the SP equation (\ref{sp}).

 It has been shown that the  SP equation (\ref{sp}) admits a  Wadati-Konno-Ichikawa  type Lax pair and is related to the sine-Gordon equation   through a chain of transformations  \cite{A.S2005}.
Then bi-Hamiltonian structure and the conservation laws were studied by   Brunelli \cite{JCB2005, JCB2006}. Sch$\ddot{a}$fer and Wayne  proved the nonexistence theorem that the equation (\ref{sp}) doesn't possess any solution representing a smooth localized pulse moving with constant shape and speed \cite{T.CE}.   Sakovich and  Sakovich also found the loop-soliton solutions of the SP equation  (\ref{sp}) \cite{A.S2006}.   Matsuno found the connection between the SP equation and the sine-Gordon equation through the hodograph transformation, and further
  found many kinds of exact solutions including    multi-soliton,  multi-loop,  multi-breather and  Periodic solutions  \cite{Matsuno, YM2008}.
    And a lot of generalizations of the SP equation has been studied, for example,  Pietrzyk, Kanatt$\breve{s}$ikov and Bandelow  introduced the vector SP equation  \cite{vsp},  a two-component SP  equation that generalizes the scalar (\ref{sp}) and describes the propagation of polarized ultra-short light pulses in cubically nonlinear anisotropic optical fibers, which can be written as
\begin{equation}
u_{n,xt}=c_{ni}u_i+c_{nijk}\left( u_iu_ju_k\right) _{xx},
\end{equation}
where $n,i,j,k = 1,2$, and the summation over the repeated indices is assumed \cite{SSvsp}. And there are
 others aspects of the SP equation have been addressed in the literature, including  integrable semi-discrete and full-discrete analogues \cite{BFF},
 well-posedness of the Cauchy problem \cite{wellG,wellD} and Riemann-Hilbert approach \cite{RHPsp}.  Using the method of testing by wave packets,   Okamoto prove the unique global existence of small solutions to the SP equation (\ref{sp}) when the small  initial data $u_0$  \cite{MO}.

In 1974,  Manakov first  carried out  the  study  on the long-time behavior of nonlinear wave equations solvable by the inverse scattering method \cite{Manakov1974}.
 Then by using this method, Zakharov and Manakov   give the first result   for large-time asymptotic  of solutions for the  NLS equation  with  decaying initial value  in 1976 \cite{ZM1976}.    The inverse scattering method    also    worked  for long-time behavior of integrable systems    such as  KdV,  Landau-Lifshitz  and the reduced Maxwell-Bloch   system \cite{SPC,BRF,Foka}.   In 1993, 
    Deift and Zhou  developed a  nonlinear steepest descent method to rigorously obtain the long-time asymptotics behavior of the solution for the MKdV equation
by deforming contours to reduce the original  Riemann-Hilbert problem (RHP) to a  model one  whose solution is calculated in terms of parabolic cylinder functions \cite{RN6}. 
Since then    this method
has been widely  applied  to  the focusing NLS equation, KdV equation, Fokas-Lenells equation,    derivative NLS equation, short-pluse equation and  Camassa-Holm equation  etc. \cite{RN9,RN10,Grunert2009,MonvelCH,xu2015,xufan2013,xusp}.

In recent years,   McLaughlin and   Miller further  presented a $\bar\partial$ steepest descent method which combine   steepest descent  with  $\bar{\partial}$-problem  rather than the asymptotic analysis
 of singular integrals on contours to analyze asymptotic of orthogonal polynomials with non-analytical weights  \cite{MandM2006,MandM2008}.
When  it  is applied  to integrable systems,   the $\bar\partial$ steepest descent method  also has  displayed some advantages,  such as   avoiding
  delicate estimates involving $L^p$ estimates  of Cauchy projection operators, and leading  the non-analyticity in the RHP reductions
   to a $\bar{\partial}$-problem in some sectors of the complex plane  which can be solved by being recast into an integral equation and by
   using Neumann series.   Dieng and  McLaughin use it to study the defocusing NLS equation  under essentially minimal regularity assumptions on finite mass
   initial data \cite{DandMNLS};  Cussagna and  Jenkins study the defocusing NLS equation with finite density initial data \cite{SandRNLS};  
   This $\bar\partial$ steepest descent method    also was successfully
    applied to prove asymptotic stability of N-soliton solutions to focusing NLS equation \cite{fNLS};
    Jenkins et.al  studied  soliton resolution for the derivative nonlinear NLS equation  for generic initial data in a weighted Sobolev space \cite{Liu3}.  Their work
    provided  the   soliton resolution property  for  derivative NLS equation, which    decomposes  the solution into the sum of a finite number of separated solitons and a
    radiative parts when $t\to\infty$.  And  the dispersive part contains two components, one coming from the continuous spectrum and another from the interaction of the discrete and continuous spectrum.

In our paper, we obtain  the  soliton resolution and  long-time asymptotic behavior  for the SP equation   (\ref{sp})
 with initial data $u_0\in H^{1,1}$   by using    $\bar\partial$ steepest descent method.
 This  paper is arranged as follows.  In section 2,  we introduce  two kinds of eigenfunctions to
 formulate the spectral  singularity of  the  Lax pair for the short-pluse equation. The analytical and asymptotics
 of the eigenfunctions are further studied.  In section 3,  following the
idea in  \cite{xusp},   we   construct  a  RH   problem   for   $M(z)$ to formulate the initial value problem of  the short-pluse equation (\ref{sp})-(\ref{sp2})   in an
alternative space variable $y$ instead of the original space variable $x$.   In section 4,   we   introduce
a function $T(z)$ to  define a new   RH problem  for  $M^{(1)}(z)$,  which  admits a regular discrete spectrum and  two  triangular  decompositions of the jump matrix
near critical point $\pm z_0$.
   In section 5,  by introducing a matrix-valued  function  $R(z)$,  we obtain  a mixed $\bar{\partial}$-RH problem  for  $M^{(2)}$  by continuous extension to $M^{(1)}$.
     In section 6,  we decompose  $M^{(2)}(z)$    into a
 model RHP  problem  for  $M^{RHP}(z)$ and a  pure $\bar{\partial}$ Problem for  $M^{(3)}(z)$.
 The  $M^{RHP}(z)$  can be obtained  via  an outer model $M^{(out)}(z)$   for the soliton components to  be solved   in Section \ref{sec7},
  and an inner model $M^{(\pm z_0)}$  for the stationary phase point $\pm z_0$   which are   approximated   by  a solvable model  for   $M^{sp}$  obtained in  \cite{xusp} in Section \ref{sec8}.
  In section \ref{sec9},  we compute  the error function  $E(z)$ with a  small-norm Riemann-Hilbert problem.
 In Section 10,   we analyze  the $\bar{\partial}$-problem  for $M^{(3)}$.
   Finally, in Section 11,   based on  the result obtained above,   a relation formula
   is found
\begin{align}
 M(z) = M^{(3)}(z) E(z) M^{out}(z)  T(z)^{-\sigma_3},\nonumber
\end{align}
from which   we then obtain the soliton  resolution  and  long-time   asymptotic behavior  for the short-pluse equation (\ref{sp}).

\section {The spectral analysis}

The SP equation (\ref{sp})  admits the Lax pair \cite{A.S2005}
\begin{equation}
\Phi_x = \left( \lambda\sigma_3+L_0\right) \Phi,\hspace{0.5cm}\Phi_t =\left(  \frac{1}{4\lambda}\sigma_3+M_0\right) \Phi, \label{lax0}
\end{equation}
where
\begin{equation}
L_0=\lambda u_x\sigma_1,\nonumber
\end{equation}
\begin{equation}
M_0=\frac{\lambda}{2}u^2\sigma_3+\left(\frac{\lambda}{2}u^2u_x+\frac{1}{2}u\right)\sigma_1,\nonumber
\end{equation}
and $\sigma_1$, $\sigma_2$ and $\sigma_3$  are  Pauli matrices
\begin{equation}
\sigma_1=\left(\begin{array}{cc}
0 & 1  \\
1 & 0
\end{array}\right),\hspace{0.5cm}\sigma_2=\left(\begin{array}{cc}
0 & -i  \\
i & 0
\end{array}\right),\hspace{0.5cm}\sigma_3=\left(\begin{array}{cc}
1 & 0   \\
0 & -1
\end{array}\right).\hspace{0.5cm}.\nonumber
\end{equation}
In order to make the presentation close to   the cases of the  CH  equation  \cite{MonvelCH}, we introduce the spectral parameter $z=i\lambda$.

First, we consider the symmetry of  the eigenfunction  $\Phi$. Note that
$$L_0(-z)=-L_0(z),\  M_0(-z)=-M_0(z), \ \sigma_2\sigma_1=-\sigma_1\sigma_2,\  \sigma_2\sigma_3=-\sigma_1\sigma_3,$$
  we find that all  $\overline{\Phi(\bar{z})}$, $\Phi(-z)$ and $\sigma_2\Phi(z)\sigma_2$  satisfy the Lax pair  (\ref{lax0}) with the  same asymptotics, which
  implies that
\begin{equation}
\overline{\Phi(\bar{z})}=\Phi(-z)=\sigma_2\Phi(z)\sigma_2.\label{symPhi}
\end{equation}

To study the long time asymptotic behaviors, usually we only use the $x$-part of Lax pair to analyze the initial value problem, and the $t$-part   is  used
to determine the time evolution of the scattering data for the integrable equations by inverse scattering transform method.
But unlike those of NLS and  derivative NLS euqations \cite{DandMNLS,SandRNLS,fNLS,YYL},    the Lax pair   (\ref{lax0}) for SP equation  has
 singularities at $z=0$ and   $z=\infty$. In order to control the behavior of solutions of (\ref{lax0}) and construct the solution $u(x,t)$ of the SP equation (\ref{sp}), we need
 use the $t$-part and  the expansion of the eigenfunction as spectral parameter $z\to 0$. So we use two different transformations respectively to analyze these two singularities $z=0$ and $z=\infty$. Because in the case $z=0$ it has well property, we first consider this case.

\noindent \textbf{Case I: z=0.}

 Consider the Jost solutions of  the Lax pair (\ref{lax0}), which are restricted by the boundary conditions
\begin{equation}
\Phi_\pm \sim  e^{i(zx-t/4z)\sigma_3}, \hspace{0.5cm}x\to \pm\infty.\label{asyx}
\end{equation}
By making transformation
\begin{equation}
\mu^0_\pm=\Phi_\pm e^{-i(zx-t/4z)\sigma_3},\label{trans2}
\end{equation}
we then  have
\begin{equation*}
\mu^0_\pm \sim I, \hspace{0.5cm} x \rightarrow \pm\infty.
\end{equation*}
Moreover,  $\mu^0_\pm$    satisfy an  equivalent Lax pair
\begin{align}
&(\mu^0_\pm)_x = iz[\sigma_3,\mu^0_\pm]+L_0\mu^0_\pm,\label{lax0.1}\\
&(\mu^0_\pm)_t = \frac{1}{4iz}[\sigma_3,\mu^0_\pm]+M_0\mu^0_\pm, \label{lax0.2}
\end{align}
from which we obtain its total differential form
\begin{equation}
d\left(e^{-i(zx-t/4z)\hat{\sigma}_3}\mu^0_\pm \right)=e^{-i(zx-t/4z)\hat{\sigma}_3}\left(L_0dx+M_0dt \right)\mu^0_\pm,
\end{equation}
whose   solutions  can be expressed as  Volterra type integrals
\begin{equation}
\mu^0_\pm=I+\int_{x}^{\pm \infty}e^{iz(x-y)\hat{\sigma}_3}L_0(y)\mu^0_\pm(y)dy\label{intmu0}.
\end{equation}
Then we can show that

\begin{proposition}
	 As $u(\cdot,t)\in H^1(\mathbb{R})$ for all $t\in \mathbb{R}$,	 the fundamental eigenfunctions $\mu^0_\pm$  exists  and is   unique.
\end{proposition}

Denote   $\mu^0_\pm=\left(\left[ \mu^0_\pm\right]_1, \left[ \mu^0_\pm\right]_2 \right) $,  where  $\left[ \mu^0_\pm\right] _1$ and $\left[ \mu^0_\pm\right] _2$ are
the first and second columns of $\mu^0_\pm$ respectively.
 Then  from  (\ref{intmu0}),   we can show that  $\left[ \mu^0_-\right] _1$ and $\left[ \mu^0_+\right] _2$ are analysis in $\mathbb{C}^-$, and $\left[ \mu^0_+\right] _1$
 and $\left[ \mu^0_-\right] _2$ are analysis in $\mathbb{C}^+$.

It is necessary to discuss the asymptotic behaviors of the Jost solutions $\mu^0_\pm$ as $z \rightarrow 0$. We consider the following  asymptotic  expansions
\begin{align}
&\mu_\pm^{0}=\mu_\pm^{0,(0)}+\mu_\pm^{0,(1)}z+\mu_\pm^{0,(2)}z^2+\mathcal{O}(z^{3}),\hspace{0.5cm}\text{as }z \rightarrow 0,\hspace{0.5cm}j=1,2,\label{expansion1}
\end{align}
where $\mu_\pm^{0,(k)}$ isn't depend on $z$, $k=0,1,2,...$.

Substituting   (\ref{expansion1})   into the Lax pair (\ref{lax0.1}) and (\ref{lax0.2}),   and comparing the coefficients,  we obtain
\begin{equation}
\mu_\pm^{0,(0)}=I,\hspace{0.5cm}\mu_\pm^{0,(1)}=izu\sigma_1,\hspace{0.5cm}\mu_\pm^{0,(2)}=-\frac{u^2}{2}I+i(u^2u_x-2u_t)\sigma_2.
\end{equation}

\noindent \textbf{Case II: $z\to \infty$.}

In order to control asymptotic behavior of the Lax pair (\ref{lax0}) as $z\to \infty$,
 we make a transformation
\begin{equation}
\Phi_\pm=G\mu_\pm e^{izp\sigma_3}\label{transmu},
\end{equation}
where
\begin{align}
&p(x,t,z)=x-\int_{x}^{+\infty} (\sqrt{m(y,t)}-1)dy-\dfrac{t}{4z^2}\label{p},\\
&m(x,t)=1+u_x(x,t)^2\label{m},\\
&G(x,t) =\sqrt{\dfrac{\sqrt{m}+1}{2\sqrt{m}}}\left(\begin{array}{cc}
1 & -\frac{\sqrt{m}-1}{u_x}   \\[4pt]
\frac{\sqrt{m}-1}{u_x} & 1
\end{array}\right).
\end{align}
Then    the SP  equation (\ref{sp})  is changed  into an   equivalent form
\begin{equation*}
(\sqrt{m})_t=\frac{1}{2}(u^2\sqrt{m})_x.
\end{equation*}
 And from (\ref{p}),  we have
\begin{equation}
p_x=\sqrt{m}, \ \ p_t=\frac{1}{2}u^2\sqrt{m}-\frac{1}{4z^2}.
\end{equation}

The Lax pair  (\ref{lax0}) is changed into  a new Lax pair
\begin{align}
&(\mu_\pm)_x = izp_x[\sigma_3,\mu_\pm]+P\mu_\pm,\label{lax1.1}\\
&(\mu_\pm)_t =izp_t[\sigma_3,\mu_\pm]+Q\mu_\pm, \label{lax1.2}
\end{align}
where
\begin{align}
&P=\dfrac{iu_{xx}}{2m}\sigma_2,\ \  Q=\frac{1}{4iz}\left(\frac{1}{\sqrt{m}} -1\right) \sigma_3+\dfrac{iu_{xx}u^2}{4m}\sigma_2-\frac{u_x}{4iz\sqrt{m}}\sigma_1,\\
&\mu_\pm \sim I, \hspace{0.5cm} x \rightarrow \pm\infty.
\end{align}
The Lax pair (\ref{lax1.1})-(\ref{lax1.2}) can be written in to a  total differential form
\begin{equation}
d\left(e^{-izp\hat{\sigma}_3}\mu_\pm \right)=e^{-izp\hat{\sigma}_3}\left(Pdx+Qdt \right)\mu_\pm,
\end{equation}
 which leads to two  Volterra type integrals
\begin{equation}
\mu_\pm=I+\int_{x}^{\pm \infty}e^{iz(p(x)-p(y))\hat{\sigma}_3}P(y)\mu_\pm(y)dy\label{intmu}.
\end{equation}
Similarly, we denote $\mu_\pm=\left(\left[ \mu_\pm\right] _1,\left[ \mu_\pm\right] _2 \right) $, then  we can show  that $\left[ \mu_-\right] _1$ and $\left[ \mu_+\right] _2$ are analysis in $\mathbb{C}^-$, and $\left[ \mu_+\right] _1$ and $\left[ \mu_-\right] _2$ are analysis in $\mathbb{C}^+$.   And the  $\mu_\pm$  admit the  asymptotics
\begin{align}
\mu_\pm=I+\dfrac{D_1}{z}+\mathcal{O}(z^{-2}),\hspace{0.5cm}z \rightarrow \infty,\label{asymu}
\end{align}
where the off-diagonal entries of the matrix $D_1(x, t)$ are
\begin{equation}
D_{12}(x,t)=D_{21}(x,t)=\dfrac{iu_{xx}}{4m\sqrt{m}}.
\end{equation}
Since   $\Phi_\pm$ are two fundamental matrix solutions of the  Lax  pair (\ref{lax0}),  there exists a linear  relation between $\Phi_+$ and $\Phi_-$, namely
\begin{equation}
\Phi_+(z;x,t)=\Phi_-(z;x,t)S(z),\hspace{0.5cm} z\in \mathbb{C},\label{scattering}
\end{equation}
where $S(z)$ is called scattering matrix which only depend on $z$.   Form the symmetry relation (\ref{symPhi})  of  $\Phi_\pm$,  the matrix   $S(z)$
also admits the symmetry
\begin{align*}
S(z) =\left(\begin{array}{cc}
\overline{a(z)} & b(z)   \\[4pt]
-\overline{b(z)} & a(z)
\end{array}\right).
\end{align*}
And   combining (\ref{transmu}) and (\ref{scattering}) gives
\begin{equation}
\mu_-(z)=\mu_+(z)e^{izp\hat{\sigma}_3}S(z),
\end{equation}
which can be written as
\begin{equation}
\left(\left[ \mu_-\right] _1,\left[ \mu_-\right] _2 \right)=\left(\left[ \mu_+\right] _1,\left[ \mu_+\right] _2 \right)\left(\begin{array}{cc}
\overline{a(z)} & e^{2izp}b(z)   \\[4pt]
-e^{-2izp}\overline{b(z)} & a(z)
\end{array}\right),
\end{equation}
which implies that
\begin{equation}
a(z)=\det\left(\left[ \mu_+\right] _1,\left[ \mu_-\right] _2 \right),\label{a}
\end{equation}
and  $a(z)$  is analytical  in $C^+$, and $a(z)=-\overline{a(-\bar{z})}$.   We introduce the reflection coefficient
\begin{equation}
r(z)=\dfrac{b(z)}{a(z)},
\end{equation}
with symmetry $r(-z)=\overline{r(\bar{z})}$. The zeros of $a(z)$ on $\mathbb{R}$  are known to
 occur and they correspond to spectral singularities \cite{RN3}.  They are excluded from our analysis in the this paper. To deal with our following work,
we assume our initial data satisfy this assumption.
\begin{Assumption}\label{initialdata}
	The initial data $u_0 (x) \in H^{1,1}(\mathbb{R})$     and it generates generic scattering data which satisfy that
	
	\textbf{1. }a(z) has no zeros on $\mathbb{R}$.
	
	\textbf{2. }a(z) only has finite number of simple zeros.
	
	\textbf{3. }a(z) and r(z) belong  $H^{1,1}(\mathbb{R})$.
\end{Assumption}
We assume that $a(z)$ has N simple zeros $z_n \in \mathbb{C}^+, n = 1, 2,.., N$, $a(z_n) = 0$. Denote $\mathcal{Z}$=$\left\lbrace z_n\right\rbrace ^N_{n=1}$ which is the set of the zeros in $\mathbb{C}^+$ of $a(z)$.
From (\ref{asymu}) and (\ref{a}),    we obtain  the asymptotic   of $a(z)$
\begin{align}
a(z)=1+\mathcal{O}(z^{-1}),\hspace{0.5cm}z \rightarrow \infty.\label{asya}
\end{align}

\quad    We can formulate a  RH problem    by defining  the matrix function $M(x,t,z)$  with  eigenfunctions $\mu_\pm$, while
the  reconstruction formula between  the solution $u(x, t)$  and   the   RH problem    can be found from the  asymptotic   of  $\mu_\pm$ as $z \to 0$.
  So we need to calculate  the relation between $\mu_\pm$ and $\mu^0_\pm$.

The relations  (\ref{trans2}) and (\ref{transmu}) implies that  there exist constant matrices  $C_\pm(z)$ satisfying
\begin{equation}
\mu_\pm(x,t,z)=G^{-1}(x,t)\mu^0_\pm e^{i(zx-t/4z)\sigma_3}C_\pm(z)e^{-izp(x,t,z)\sigma_3},
\end{equation}
which means $\mu_\pm(x,t,z)$ exists and is unique.
Take $x\to \pm \infty$,  we have
\begin{equation}
C_+=I,\hspace{0.5cm}C_-=e^{izc\sigma_3},
\end{equation}
where
\begin{equation}
c=\int_{R}(\sqrt{m}-1) dy\label{c}
\end{equation}
is a conserved quantity under the dynamics governed by (\ref{sp}). Then we have
\begin{equation}
\mu_\pm(x,t,z)=G^{-1}(x,t)\mu^0_\pm e^{-iz\int_{\pm\infty}^x (\sqrt{m}-1) dy\sigma_3}.
\end{equation}

Since  tr$(iz\sigma_3+L_0)$=tr$(\frac{1}{4iz}\sigma_3+M_0)$=0,   by  the  Able formula,   it holds that
\begin{equation}
\det(\Phi_\pm)_x=\det(\Phi_\pm)_t=0,
\end{equation}
which together with $\det(\mu^0_\pm)=\det(\Phi_\pm)$
leads to
\begin{equation}
\det(\mu^0_\pm)_x=\det(\mu^0_\pm)_t=0,
\end{equation}
and
\begin{equation}
1=\det(\mu^0_\pm)=\det(\Phi_\pm)=\det(S(z)).
\end{equation}
Then we have $|a(z)|^2+|b(z)|^2=1$, which is equivalent to $1+|r(z)|^2=\frac{1}{|a(z)|^2}$. In the absence of spectral singularities (real zeros of $a(z)$), there also exist $\nu\in (0,1)$ such that $\nu<|a(z)|<1/\nu$ for $z\in R$, which implies $1+|r(z)|>\nu^2>0$ for $z\in R$. And from the asymptotic   of the Jost solutions $\mu^0_\pm$ as $z \rightarrow 0$, we get  the asymptotic of $a(z)$
\begin{equation}
a(z)=1+icz-\frac{c^2}{2}z^2+\mathcal{O}(z^3),\ \ z \rightarrow 0,\label{asya0}
\end{equation}
where $c$ is defined in (\ref{c}).

\section{The construction of   a RH problem}

\quad   Suppose that  $\mathcal{Z}=\{z_n, \   n=1, \cdots, N\}$  are simple zeros for $a(z)$,   we  first calculate residue conditions.
Since   $(\Phi_+^1,\Phi_+^2)$ and  $(\Phi_-^1,\Phi_-^2)$  are linearly dependent,  there exists a constant $b_k$  such that
\begin{equation*}
(\Phi_+^1,\Phi_+^2)=b_k(\Phi_-^1,\Phi_-^2),
\end{equation*}
which implies that
\begin{equation}
\left[ \mu_+\right] _1(z_k)=b_ke^{2iz_kp(z_k)}\left[ \mu_-\right] _2(z_k).
\end{equation}
 We  denote   norming constant    $c_k=b_k/a'(z_k)$,   and  the collection
$\sigma_d=  \left\lbrace z_k,c_k\right\rbrace^N_{k=1}  $
is called the \emph{scattering data}.

We define  a   sectionally meromorphic matrix
\begin{equation}
N(z;x,t)=\left\{ \begin{array}{ll}
\left( \left[ \mu_+\right] _1, a(z)^{-1}\left[ \mu_-\right] _2\right),   &\text{as } z\in \mathbb{C}^+,\\[12pt]
\left( \overline{a(\bar{z})}^{-1}\left[ \mu_-\right] _1,\left[ \mu_+\right] _2\right)  , &\text{as }z\in \mathbb{C}^-,\\
\end{array}\right.
\end{equation}
which solves the following RHP.

\noindent\textbf{RHP1}.  Find a matrix-valued function $N(z;x,t)$ which satisfies:

$\blacktriangleright$ Analyticity: $N(z;x,t)$ is meromorphic in $\mathbb{C}\setminus \mathbb{R}$ and has single poles;

$\blacktriangleright$ Symmetry: $\overline{N(\bar{z})}$=$N(-z)$=$\sigma_2N(z)\sigma_2$;

$\blacktriangleright$ Jump condition: $N$ has continuous boundary values $N_\pm$ on $\mathbb{R}$ and
\begin{equation}
N^+(z;x,t)=N^-(z;x,t)V(z),\hspace{0.5cm}z \in \mathbb{R},
\end{equation}
where
\begin{equation}
V(z)=\left(\begin{array}{cc}
1 & e^{2izp}r(z)\\
e^{-2izp}\overline{r(z)} & 1+|r(z)|^2
\end{array}\right);
\end{equation}

$\blacktriangleright$ Asymptotic behaviors:
\begin{align}
&N(z;x,t) = I+\mathcal{O}(z^{-1}),\hspace{0.5cm}z \rightarrow \infty;
\end{align}

$\blacktriangleright$ Residue conditions: N has simple poles at each point in $ \mathcal{Z}\cup \bar{\mathcal{Z}}$ with:
\begin{align}
&\res_{z=z_n}N(z)=\lim_{z\to z_n}N(z)\left(\begin{array}{cc}
0 & c_ne^{-2iz_kp(z_k)}\\
0 & 0
\end{array}\right),\\
&\res_{z=\bar{z}_n}N(z)=\lim_{z\to \bar{z}_n}N(z)\left(\begin{array}{cc}
0 & 0\\
-\bar{c}_ne^{-2i\bar{z}_np(\bar{z}_n)} & 0
\end{array}\right).
\end{align}

We denote
\begin{equation}
c_+(x,t)=\int_{x}^{+\infty}(\sqrt{m(k,t)}-1)dk,\label{c+}
\end{equation}
and  consider the asymptotic  of $N(z;x,t)$
\begin{equation}
N(z;x,t)=G^{-1}(x,t)\left[I+z(ic_+\sigma_3+iu\sigma_1+\mathcal{O}(z^2)) \right], \ \ z\to 0.\label{asyM0}
\end{equation}
from which  it is difficult to reconstruct the solution of the SP equation (\ref{sp}), since  $p(x,t,z)$ is still unknown.
To overcome this,  we introduce a  new   scale
\begin{equation}
y(x,t)=x-\int_{x}^{+\infty} \left(\sqrt{m(k,t)}-1\right) dk=x-c_+(x,t).
\end{equation}

The price to pay for this is that the solution of the initial problem can be given only implicitly,
or parametrically: it will be given in terms of functions in the new scale, whereas the original scale will also be given in terms of functions in the new scale.
By the definition of the new scale $y(x, t)$, we define
\begin{equation}
M(y,t,z)=N(x(y,t),t,z),
\end{equation}
which satisfies the following  RH problem.

\noindent\textbf{RHP2}. Find a matrix-valued function  $M(z)=M(y,t,z)$ which satisfies:

$\blacktriangleright$ Analyticity: $M(z)$ is meromorphic in $\mathbb{C}\setminus R$ and has single poles;

$\blacktriangleright$ Symmetry: $\overline{M(\bar{z})}$=$M(-z)$=$\sigma_2M(z)\sigma_2$;

$\blacktriangleright$ Jump condition: M has continuous boundary values $M_\pm$ on $R$ and
\begin{equation}
M^+(z)=M^-(z)V(z),\hspace{0.5cm}z \in R,\label{jump}
\end{equation}
where
\begin{equation}
V(z)=\left(\begin{array}{cc}
1 & e^{2i(zy-\frac{t}{4z})}r(z)\\
e^{-2i(zy-\frac{t}{4z})}\overline{r(z)} & 1+|r(z)|^2
\end{array}\right);\label{jumpv}
\end{equation}

$\blacktriangleright$ Asymptotic behaviors:
\begin{align}
&M(z) = I+\mathcal{O}(z^{-1}),\hspace{0.5cm}z \rightarrow \infty;
\end{align}

$\blacktriangleright$ Residue conditions: M has simple poles at each point in $ \mathcal{Z}\cup \bar{\mathcal{Z}}$ with:
\begin{align}
&\res_{z=z_n}M(z)=\lim_{z\to z_n}M(z)\left(\begin{array}{cc}
0 & c_ne^{-2i(z_ny-\frac{t}{4z_n})}\\
0 & 0
\end{array}\right),\label{RES1}\\
&\res_{z=\bar{z}_n}M(z)=\lim_{z\to \bar{z}_n}M(z)\left(\begin{array}{cc}
0 & 0\\
-\bar{c}_ne^{-2i\bar{z}_n(\bar{z}_ny-\frac{t}{4\bar{z}_n})} & 0
\end{array}\right).\label{RES2}
\end{align}
From the asymptotic behavior of the functions $\mu_\pm$ and (\ref{asyM0}), we have following reconstruction formula of $u(x,t)=u(y(x,t),t)$:
\begin{equation}
u(x,t)=u(y(x,t),t)=\lim_{z\to 0}\dfrac{\left(M(y,t,0)^{-1}M(y,t,z) \right)_{12} }{iz},\label{recons u}
\end{equation}
where
\begin{equation}
x(y,t)=y+c_+(x,t)=y+\lim_{z\to 0}\dfrac{\left(M(y,t,0)^{-1}M(y,t,z) \right)_{11}-1 }{iz}.
\end{equation}

\section{Conjugation}

\quad
In the jump matrix (\ref{jumpv}), we  denote  the  oscillatory term
\begin{equation}
e^{2i(zy-\frac{t}{4z})}=e^{2it \theta(z) },\ \ \theta(z)=z\frac{y}{t}-\frac{1}{4z},\label{theta}
\end{equation}
it will be found that  the long-time asymptotic  of RHP2  is affected by the growth and decay of the exponential function $e^{2it\theta}$ appearing in
 both the jump relation and the residue conditions.
 In this section, we introduce  a new transform  $M(z)\to M^{(1)}(z)$, from which we make that the  $M^{(1)}(z)$
 is well behaved as $|t|\to \infty$ along any characteristic line.

  Let $\xi=\frac{y}{t}<0$, then  $z_0=\sqrt{-\frac{1}{4\xi}}\in \mathbb{R}$,  where $\pm z_0$ are the two critical points of the phase function $\theta(z)$.
  The case of $\xi>0$ is discussed by   Xu   \cite{xusp}.
  Then (\ref{theta}) can be written as
\begin{equation}
\theta(z)=-\frac{z}{4}\left(\frac{1}{z_0^2} +\frac{1}{z^2} \right) ,\hspace{0.5cm}\text{Re} (2it\theta)=-2t\text{Im}\theta=-2t\text{Im}z\left(\xi+\frac{1}{4|z|^2} \right) .
\end{equation}
The partition $\Delta^\pm_{z_0,\eta}$ of $\left\lbrace 1,...,N\right\rbrace $ for $z_0 \in \mathbb{R}$, $\eta$ = ${\rm sgn}(t)$ is defined as follow:
\begin{align*}
\Delta^+_{z_0,1}=\Delta^-_{z_0,-1}=\left\lbrace k \in \left\lbrace 1,...,N\right\rbrace  ||z_k|<z_0\right\rbrace, \\
\Delta^-_{z_0,1}=\Delta^+_{z_0,-1}=\left\lbrace k \in \left\lbrace 1,...,N\right\rbrace  ||z_k|>z_0\right\rbrace.
\end{align*}
This partition splits the residue coefficients $c_n$ in two sets which is shown in Figure. \ref{fig1}.
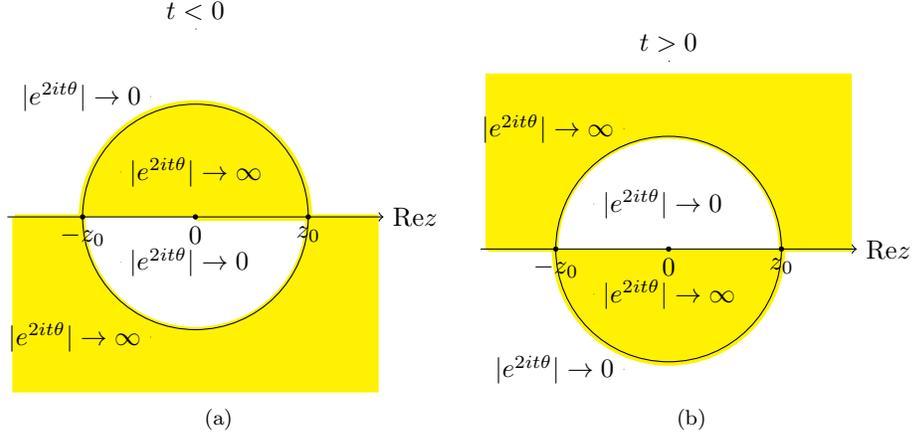
\begin{figure}[H]
	\centering
	\centering
	\subfigure[]{
		\begin{tikzpicture}[node distance=2cm]
		\filldraw[yellow,line width=3] (0,0) --(1.5,0)arc(0:180:1.5);
		\filldraw[yellow,line width=2] (-2.4,0)--(-2.4,-2.3)--(2.4,-2.3)--(2.4,0)--(1.5,0)--(1.5,0) arc(360:180:1.5)--(-1.5,0)--(-2.4,0);
		\draw[->](-2.5,0)--(2.5,0)node[right]{Re$z$};
		\draw[](0,2.5)--(0,2.5)node[above]{$t<0$};
		\draw(0,0) circle (1.5);
		\coordinate (A) at (-1,0.6);
		\coordinate (B) at (-1,-0.6);
		\coordinate (G) at (-0.6,1.6);
		\coordinate (H) at (-0.6,-1.6);
		\coordinate (I) at (0,0);
		\coordinate (C) at (1.5,0);
		\coordinate (D) at (-1.5,0);
		\fill (A) circle (0pt) node[right] {$|e^{2it\theta}|\to \infty$};
		\fill (B) circle (0pt) node[right] {$|e^{2it\theta}|\to 0$};
		\fill (G) circle (0pt) node[left] {$|e^{2it\theta}|\to 0$};
		\fill (H) circle (0pt) node[left] {$|e^{2it\theta}|\to \infty$};
		\fill (I) circle (1pt) node[below] {$0$};
		\fill (C) circle (1pt) node[below] {$z_0$};
		\fill (D) circle (1pt) node[below] {$-z_0$};
		\label{zplane1}
		\end{tikzpicture}
	}
	\subfigure[]{
		\begin{tikzpicture}[node distance=2cm]
		\filldraw[yellow,line width=3] (-1.5,0) arc(180:360:1.5);
		\filldraw[yellow,line width=2] (-2.4,0)--(-2.4,2.3)--(2.4,2.3)--(2.4,0)--(1.5,0)--(1.5,0) arc(0:180:1.5)--(-1.5,0)--(-2.4,0);
		\draw[->](-2.5,0)--(2.5,0)node[right]{Re$z$};
		\draw[](0,2.5)--(0,2.5)node[above]{$t>0$};
		\draw(0,0) circle (1.5);
		\coordinate (A) at (-1,0.6);
		\coordinate (B) at (-1,-0.6);
		\coordinate (G) at (-0.6,1.6);
		\coordinate (H) at (-0.6,-1.6);
		\coordinate (I) at (0,0);
		\coordinate (C) at (1.5,0);
		\coordinate (D) at (-1.5,0);
		\fill (A) circle (0pt) node[right] {$|e^{2it\theta}|\to 0$};
		\fill (B) circle (0pt) node[right] {$|e^{2it\theta}|\to \infty$};
		\fill (G) circle (0pt) node[left] {$|e^{2it\theta}|\to \infty$};
		\fill (H) circle (0pt) node[left] {$|e^{2it\theta}|\to 0$};
		\fill (I) circle (1pt) node[below] {$0$};
		\fill (C) circle (1pt) node[below] {$z_0$};
		\fill (D) circle (1pt) node[below] {$-z_0$};
		\label{zplane2}
		\end{tikzpicture}
	}
	\caption{In the yellow region, $|e^{2it\theta}|\to \infty$ when $t\to\pm\infty$ respectively. And in white region, $|e^{2it\theta}|\to 0$ when $t\to\pm\infty$ respectively.}
	\label{fig1}
\end{figure}
We define the  following functions and notation which will used later
\begin{align}
&k(s)=-\frac{1}{2\pi}\log(1+|r(s)|^2),\label{ks}\\
&I_+=\left( -\infty,-z_0\right] \cup\left[ z_0,+\infty\right) , \hspace{0.5cm}I_-=[-z_0,z_0]\\
&\delta (z)=\delta (z,z_0,\eta)=\exp\left(i\int _{I_\eta}\dfrac{k(s)ds}{s-z}\right)\\
&T(z)=T(z,z_0,\eta)=\prod_{k\in \Delta^-_{z_0,\eta}}\dfrac{z-\bar{z}_k}{z-z_k}\delta (z) \label{T}, \\
&\beta^\pm(z, z_0,\eta)=-\eta k(\pm z_0)\log(\eta(z\mp z_0+1))+\int _{I_\eta}\dfrac{k(s)-X_{\eta,\pm}(s)k(\pm z_0)}{s-z}ds,\\
&T_0(\pm z_0)=T(\pm z_0,z_0,\eta)=\prod_{k\in \Delta^-_{z_0,\eta}}\dfrac{\pm z_0-\bar{z}_k}{\pm z_0-z_k}e^{i\beta^\pm (z_0,\pm z_0,\eta)},\label{T0}
\end{align}
where  $X_{\eta,+}(s)$ and $X_{\eta,-}(s)$ are the characteristic functions of the interval $\eta z_0<\eta s<\eta z_0+1$ and $-\eta z_0-1<\eta s<-\eta z_0$ respectively. In all of the above formulas, we choose the principal branch of power and logarithm functions.

\begin{proposition}\label{proT}
	The function defined by (\ref{T}) has following properties:\\
	(a) $T$ is meromorphic in $C\setminus I_\eta$, for each $n\in\Delta^-_{z_0,\eta}$, $T(z)$ has a simple pole at $z_n$ and a simple zero at $\bar{z}_n$;\\
	(b) For $z\in C\setminus I_\eta$, $\overline{T(\bar{z})}T(z)=1$;\\
	(c) For $z\in I_\eta$, as z approaches the real axis from above and below, $T$ has boundary values $T_\pm$, which satisfy:
	\begin{equation}
	T_+(z)=(1+|r(z)|^2)T_-(z),\hspace{0.5cm}z\in I_\eta;
	\end{equation}
	(d) As $|z|\to \infty$ with $|arg(z)|\leq c<\pi$,
	\begin{equation}
	T(z)=1+\frac{i}{z}\left[ 2\sum_{k\in \Delta^-_{z_0,\eta}} Im(z_k)-\int _{I_\eta}k(s)ds\right]+ \mathcal{O}(z^{-2})\label{expT};
	\end{equation}
	(e) $T(z)$ is continuous at $z=0$, and as $|z|\to 0$,
	\begin{equation}
	T(z)=T(0)\left( 1+zT_1\right) +	\mathcal{O}(z^2) \label{expT0},
	\end{equation}
	where
	\begin{equation}
	T_1= 2\sum_{k\in \Delta^-_{z_0,\eta}} \frac{Im(z_k)}{z_k}-\int _{I_\eta}\frac{k(s)}{s^2}ds\label{T1};
	\end{equation}
	(f) As $z\to \pm z_0$, along $z=\pm z_0+e^{i\psi}l$, $l>0$, $|\psi|\leq c<\pi$,
	\begin{equation}
	|T(z,z_0,\eta)-T_0(\pm z_0,\eta)(\eta(z\mp z_0))^{i\eta k(\pm z_0)}|\leq C|z\mp z_0|^{1/2}.
	\end{equation}
\end{proposition}
\begin{proof}
	 The proof of  above properties  can be obtain by simple calculation, for details,  see   \cite{YYL}.
\end{proof}

We now use $T(z)$ to define a new  matrix-valued   function $M^{(1)}(z)$
\begin{equation}
M^{(1)}(z)=M(z)T(z)^{-\sigma_3},\label{transm1}
\end{equation}
which then satisfies the following RH problem.

\noindent \textbf{RHP3}. Find a matrix-valued function  $  M^{(1)}(z )$ which satisfies:

$\blacktriangleright$ Analyticity: $M^{(1)}(z )$ is meromorphic in $\mathbb{C}\setminus \mathbb{R}$ and has single poles;

$\blacktriangleright$ Symmetry: $\overline{M^{(1)}(\bar{z})}$=$M^{(1)}(-z)$=$\sigma_2M^{(1)}(z)\sigma_2$;

$\blacktriangleright$ Jump condition: $M^{(1)}$ has continuous boundary values $M^{(1)}_\pm$ on $\mathbb{R}$ and
\begin{equation}
M^{(1)}_+(z )=M^{(1)}_-(z )V^{(1)}(z),\hspace{0.5cm}z \in \mathbb{R},\label{jump3}
\end{equation}
where
\begin{align}
&\text{as } z\in R\setminus I_\eta, V^{(1)}(z)=
\left(\begin{array}{cc}
1 & 0\\
\overline{r(z)}T(z)^{2}e^{-2it\theta} & 1
\end{array}\right)\left(\begin{array}{cc}
1 & r(z)T(z)^{-2}e^{2it\theta}\\
0 & 1
\end{array}\right), \\
&\text{as }z\in I_\eta\setminus \left\lbrace \pm z_0\right\rbrace , V^{(1)}(z)=\left(\begin{array}{cc}
1 & \dfrac{r(z)T_-(z)^{-2}}{1+|r(z)|^2}e^{2it\theta}\\
0 & 1
\end{array}\right)\left(\begin{array}{cc}
1 & 0\\
\dfrac{\overline{r(z)}T_+(z)^{2}}{1+|r(z)|^2}e^{-2it\theta} & 1
\end{array}\right);
\end{align}

$\blacktriangleright$ Asymptotic behaviours:
\begin{align}
&M^{(1)}(z )= I+\mathcal{O}(z^{-1}),\hspace{0.5cm}z \rightarrow \infty;\label{asymbehv4}
\end{align}

$\blacktriangleright$ Residue conditions: $M^{(1)}$ has simple poles at each point in $ \mathcal{Z}\bigcup \bar{\mathcal{Z}}$ with:

\hspace{1cm}  For $n\in \Delta^+_{z_0,\eta}$,
\begin{align}
&\res_{z=z_n}M^{(1)}(z)=\lim_{z\to z_n}M^{(1)}(z)\left(\begin{array}{cc}
0 & c_nT(z_n)^{-2}e^{-2i\theta_n t}\\
0 & 0
\end{array}\right),\label{RES7}\\
&\res_{z=\bar{z}_n}M^{(1)}(z)=\lim_{z\to \bar{z}_n}M^{(1)}(z)\left(\begin{array}{cc}
0 & 0\\
\bar{c}_nT(\bar{z}_n)^2e^{2i\bar{\theta}_n t} & 0
\end{array}\right).\label{RES8}
\end{align}

\hspace{1cm} For $n\in \Delta^-_{z_0,\eta}$,
\begin{align}
&\res_{z=z_n}M^{(1)}(z)=\lim_{z\to z_n}M^{(1)}(z)\left(\begin{array}{cc}
0 & 0\\
c_n(1/T)'(z_n)^{-2}e^{2i\theta_n t}& 0
\end{array}\right),\label{RES9}\\
&\res_{z=\bar{z}_n}M^{(1)}(z)=\lim_{z\to \bar{z}_n}M^{(1)}(z)\left(\begin{array}{cc}
0 & \bar{c}_nT'(\bar{z}_k)^{-2}e^{-2i\bar{\theta}_n t}\\
0 & 0
\end{array}\right).\label{RES10}
\end{align}
where we denote $\theta_n=\theta(z_n).$

\begin{proof}
	The analyticity, jump condition and asymptotic behaviours of $M^{(1)}(z)$ is directly from its definition, the proposition \ref{proT} and the properties of $M$. As for residues, because $T(z)$ is analytic at each $z_n$ and $\bar{z}_n$ for $n\in \Delta^+_{z_0,\eta}$, from (\ref{RES1}), (\ref{RES2}) and (\ref{transm1}) we obtain residue conditions at these point immediately.

For $n\in \Delta^-_{z_0,\eta}$, we denote $M(z)=\left(M_1(z), M_2(z) \right) $, then
$$M^{(1)}(z)=\left(M^{(1)}_1(z), M^{(1)}_2(z) \right) =   \left(M_1(z)T(z), M_2(z)T(z)^{-1} \right).$$
 $T(z)$ has a simple zero at $\bar{z}_n$ and a pole at $z_n$, so $z_n$ is no longer the pole of $M^{(1)}_1(z)$ with $\bar{z}_n$ becoming the pole of it. And $M^{(1)}_2(z)$ has opposite situation. It has pole at $z_n$ and a removable singularity at $\bar{z}_n$. The calculation of it is similar as it in \cite{YYL}.
\end{proof}

\section{ A mixed $\bar{\partial}$-RH problem }

\quad  In this section,  we make continuous extension to the jump matrix $V^{(1)}$, for this purpose, we   introduce
  new contours defined as follow:
\begin{align}
&\Sigma_k=z_0+e^{(2k-1)i\pi/4}R_+,\hspace{0.5cm}k=1,4,\\
&\Sigma_k=z_0+e^{(2k-1)i\pi/4}h,\hspace{0.5cm}h\in(0,(\sqrt{2})^{-1}z_0)  ,\hspace{0.5cm}k=2,3,\\
&\Sigma_k=z_0+e^{(2k-1)i\pi/4}h,\hspace{0.5cm}h\in(0,(\sqrt{2})^{-1}z_0)  ,\hspace{0.5cm}k=5,8,\\
&\Sigma_k=-z_0+e^{(2k-1)i\pi/4}R_+,\hspace{0.5cm}k=6,7,\\
&\Sigma_k=e^{(2k-1)i\pi/4}h,\hspace{0.5cm}h\in(0,(\sqrt{2})^{-1}z_0) ,\hspace{0.5cm}k=9,10,11,12,\\
&\Sigma^{(2)}=\Sigma_1\cup\Sigma_2...\cup\Sigma_{12},
\end{align}
then the contour $\Sigma^{(2)}$ and real axis  $\mathbb{R}$   separate complex plane  $\mathbb{C}$  into ten open sectors denoted  by $\Omega_k$, $k=1,...,10$,
 starting with sector $\Omega_1$ between $I_\eta$ and $\Sigma_1$ and numbered consecutively continuing counterclockwise for $\eta$ = 1 ($\eta$ = -1 is similarly) as shown in Figure \ref{figR2}.
\begin{figure}
	\centering
	\subfigure[]{
		\begin{tikzpicture}[node distance=2cm]
		\draw[yellow, fill=yellow] (0,0)--(-1,1)--(-2,0)--(-1,-1)--(0,0);
		\draw[yellow, fill=yellow] (-2,0)--(-3,1)--(-4,0)--(-3,-1)--(-2,0);
		\draw[yellow, fill=yellow] (0,0)--(2,2)--(2.2,2)--(2.2,-2)--(2,-2)--(0,0);
		\draw[yellow, fill=yellow] (-4,0)--(-6,2)--(-6.2,2)--(-6.2,-2)--(-6,-2)--(-4,0);
		\draw[](-2,2.7)node[above]{$\eta=-1$};
		\draw(0,0)--(2,2)node[left]{$\Sigma_1$};
		\draw(0,0)--(-1,1)node[right]{$\Sigma_2$};
		\draw(0,0)--(-1,-1)node[right]{$\Sigma_3$};
		\draw(0,0)--(2,-2)node[right]{$\Sigma_4$};
		\draw[dashed](-6,0)--(2.2,0)node[right]{\scriptsize Re$z$};
		\draw(-2,0)--(-1,1)node[left]{$\Sigma_9$};
		\draw(-2,0)--(-1,-1)node[left]{$\Sigma_{12}$};
		\draw(-2,0)--(-3,-1)node[right]{$\Sigma_{11}$};
		\draw(-2,0)--(-3,1)node[right]{$\Sigma_{10}$};
		\draw(-4,0)--(-3,1)node[left]{$\Sigma_5$};
		\draw(-4,0)--(-3,-1)node[left]{$\Sigma_8$};
		\draw(-4,0)--(-6,2)node[right]{$\Sigma_6$};
		\draw(-4,0)--(-6,-2)node[right]{$\Sigma_7$};
		\draw[-latex](-1,1)--(-1.5,0.5);
		\draw[-latex](-1,-1)--(-1.5,-0.5);
		\draw[-latex](-2,0)--(-2.5,-0.5);
		\draw[-latex](-2,0)--(-2.5,0.5);
		\draw[-latex](-3,1)--(-3.5,0.5);
		\draw[-latex](-3,-1)--(-3.5,-0.5);
		\draw[-latex](-6,2)--(-5,1);
		\draw[-latex](-6,-2)--(-5,-1);
		\draw[-latex](0,0)--(-0.5,-0.5);
		\draw[-latex](0,0)--(-0.5,0.5);
		\draw[-latex](0,0)--(1,1);
		\draw[-latex](0,0)--(1,-1);
		\coordinate (A) at (1,0.5);
		\coordinate (B) at (1,-0.5);
		\coordinate (G) at (-1,0.5);
		\coordinate (H) at (-1,-0.5);
		\coordinate (I) at (0,0);
		\coordinate (C) at (-1,0.1);
		\coordinate (K) at (-4,0);
		\coordinate (l) at (-2,0);	
		\fill (C) circle (0pt) node[above] {$\Omega_3$};
		\coordinate (E) at (-2,1.2);
		\fill (E) circle (0pt) node[above] {$\Omega_2$};
		\coordinate (D) at (1,0.1);
		\fill (D) circle (0pt) node[above] {$\Omega_1$};
		\coordinate (F) at (1,-0.1);
		\fill (F) circle (0pt) node[below] {$\Omega_6$};
		\coordinate (J) at (-2,-1.2);
		\fill (J) circle (0pt) node[below] {$\Omega_5$};
		\coordinate (k) at (-1,-0.1);
		\fill (k) circle (0pt) node[below] {$\Omega_4$};
		\fill (I) circle (1pt) node[below] {$z_0$};
		\fill (K) circle (1pt) node[below] {$-z_0$};
		\fill (l) circle (1pt) node[below] {$0$};
		\coordinate (a) at (-5,0.1);
		\fill (a) circle (0pt) node[above] {$\Omega_7$};
		\coordinate (b) at (-3,0.1);
		\fill (b) circle (0pt) node[above] {$\Omega_8$};
		\coordinate (c) at (-3,-0.1);
		\fill (c) circle (0pt) node[below] {$\Omega_9$};
		\coordinate (d) at (-5,-0.1);
		\fill (d) circle (0pt) node[below] {$\Omega_{10}$};
		\end{tikzpicture}
	}
	\subfigure[]{
		\begin{tikzpicture}[node distance=2cm]
		\draw[yellow, fill=yellow] (0,0)--(-1,1)--(-2,0)--(-1,-1)--(0,0);
		\draw[yellow, fill=yellow] (-2,0)--(-3,1)--(-4,0)--(-3,-1)--(-2,0);
		\draw[yellow, fill=yellow] (0,0)--(2,2)--(2.2,2)--(2.2,-2)--(2,-2)--(0,0);
		\draw[yellow, fill=yellow] (-4,0)--(-6,2)--(-6.2,2)--(-6.2,-2)--(-6,-2)--(-4,0);
		\draw[](-2,2.7)node[above]{$\eta=+1$};
		\draw(0,0)--(2,2)node[left]{$\Sigma_1$};
	\draw(0,0)--(-1,1)node[right]{$\Sigma_2$};
	\draw(0,0)--(-1,-1)node[right]{$\Sigma_3$};
	\draw(0,0)--(2,-2)node[right]{$\Sigma_4$};
	\draw[dashed](-6,0)--(2.2,0)node[right]{\scriptsize Re$z$};
	\draw(-2,0)--(-1,1)node[left]{$\Sigma_9$};
	\draw(-2,0)--(-1,-1)node[left]{$\Sigma_{12}$};
	\draw(-2,0)--(-3,-1)node[right]{$\Sigma_{11}$};
	\draw(-2,0)--(-3,1)node[right]{$\Sigma_{10}$};
	\draw(-4,0)--(-3,1)node[left]{$\Sigma_5$};
	\draw(-4,0)--(-3,-1)node[left]{$\Sigma_8$};
	\draw(-4,0)--(-6,2)node[right]{$\Sigma_6$};
	\draw(-4,0)--(-6,-2)node[right]{$\Sigma_7$};
	\draw[-latex](-1,1)--(-1.5,0.5);
	\draw[-latex](-1,-1)--(-1.5,-0.5);
	\draw[-latex](-2,0)--(-2.5,-0.5);
	\draw[-latex](-2,0)--(-2.5,0.5);
	\draw[-latex](-3,1)--(-3.5,0.5);
	\draw[-latex](-3,-1)--(-3.5,-0.5);
	\draw[-latex](-6,2)--(-5,1);
	\draw[-latex](-6,-2)--(-5,-1);
	\draw[-latex](0,0)--(-0.5,-0.5);
	\draw[-latex](0,0)--(-0.5,0.5);
	\draw[-latex](0,0)--(1,1);
	\draw[-latex](0,0)--(1,-1);
		\coordinate (A) at (1,0.5);
		\coordinate (B) at (1,-0.5);
		\coordinate (G) at (-1,0.5);
		\coordinate (H) at (-1,-0.5);
		\coordinate (I) at (0,0);
		\coordinate (C) at (-1,0.1);
		\coordinate (K) at (-4,0);
		\coordinate (l) at (-2,0);	
		\fill (C) circle (0pt) node[above] {$\Omega_1$};
		\coordinate (E) at (-2,1.2);
		\fill (E) circle (0pt) node[above] {$\Omega_2$};
		\coordinate (D) at (1,0.1);
		\fill (D) circle (0pt) node[above] {$\Omega_3$};
		\coordinate (F) at (1,-0.1);
		\fill (F) circle (0pt) node[below] {$\Omega_4$};
		\coordinate (J) at (-2,-1.2);
		\fill (J) circle (0pt) node[below] {$\Omega_5$};
		\coordinate (k) at (-1,-0.1);
		\fill (k) circle (0pt) node[below] {$\Omega_6$};
		\fill (I) circle (1pt) node[below] {$z_0$};
		\fill (K) circle (1pt) node[below] {$-z_0$};
		\fill (l) circle (1pt) node[below] {$0$};
		\coordinate (a) at (-5,0.1);
		\fill (a) circle (0pt) node[above] {$\Omega_8$};
		\coordinate (b) at (-3,0.1);
		\fill (b) circle (0pt) node[above] {$\Omega_7$};
		\coordinate (c) at (-3,-0.1);
		\fill (c) circle (0pt) node[below] {$\Omega_{10}$};
		\coordinate (d) at (-5,-0.1);
		\fill (d) circle (0pt) node[below] {$\Omega_9$};
		\end{tikzpicture}
	}
	\caption{In the yellow region, $  R^{(2)}\not=I$,  in white region, $ R^{(2)}=I$.}
	\label{figR2}
\end{figure}
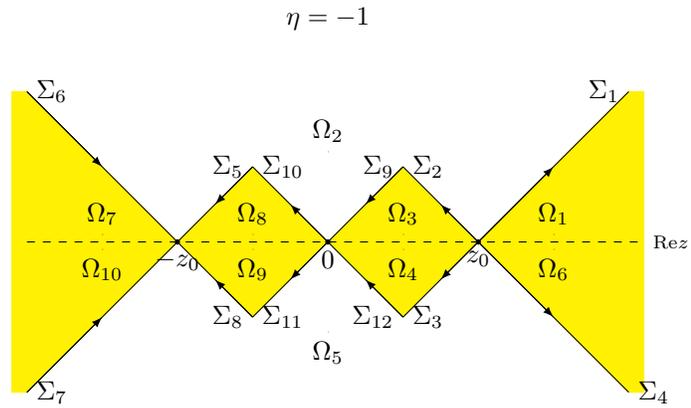
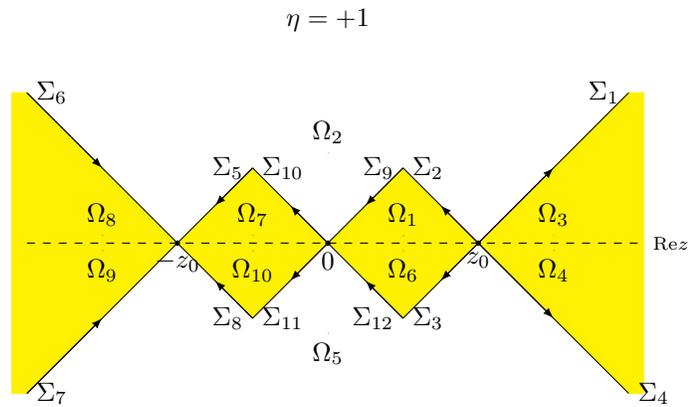

Additionally, let
\begin{equation}
\mu=\frac{1}{2}\min_{\lambda\neq\gamma\in \mathcal{Z}\cup \bar{\mathcal{Z}}}|\lambda -\gamma|.
\end{equation}
Since  there is no pole on the real axis, it holds that  $ {\rm dist}(\mathcal{Z},\mathbb{R})>\mu$.  Then we define $X_\mathcal{Z} \in C_0^\infty(\mathbb{C},[0,1])$ which only supported on the
neighborhood of $\mathcal{Z}\cup \bar{\mathcal{Z}}$,
\begin{equation}
X_\mathcal{Z}(z)=\Bigg\{\begin{array}{ll}
1 &\text{dist}(z,\mathcal{Z}\cup \bar{\mathcal{Z}})<\mu/3\\
0 &\text{dist}(z,\mathcal{Z}\cup \bar{\mathcal{Z}})>2\mu/3.\\
\end{array}
\end{equation}
In order to deform the contour $\mathbb{R}$ to the contour $\Sigma^{(2)}$, we introduce a new unknown function $M^{(2)}$ as follow:
\begin{equation}
M^{(2)}(z)=M^{(1)}(z)R^{(2)}(z),\label{transm2}
\end{equation}
where  $R^{(2)}(z)$ is chosen to satisfy the following conditions:  First, $M^{(2)}$ has no jump on the real axis, so we choose the boundary values of $R^{(2)}(z)$  through the factorization of $V^{(1)}(z)$ in (\ref{jump3}) where the new jumps on $\Sigma^{(2)}$ match a well known model RH problem; Second, we need to control the norm of $R^{(2)}(z)$, so that the $\bar{\partial}$-contribution to the long-time asymptotics of $u(y,t)$ can be ignored; Third the residues are unaffected by the transformation.  So we choose $R^{(2)}(z)$ as
\begin{equation}
R^{(2)}(z)=\left\{\begin{array}{lll}
\left(\begin{array}{cc}
1 & (-1)^{m_j}R_j(z)e^{2it\theta}\\
0 & 1
\end{array}\right), & z\in \Omega_j,j=1,4,7,9;\\
\\
\left(\begin{array}{cc}
1 & 0\\
(-1)^{m_j}R_j(z)e^{-2it\theta} & 1
\end{array}\right),  &z\in \Omega_j,j=3,6,8,10;\\
\\
I  &z\in \Omega_2\cup\Omega_5;\\
\end{array}\right.
\end{equation}
where $m_1=m_3=m_7=m_8=1$, $m_4=m_6=m_9=m_{10}=0$, and the function $R_j$, $j=1,3,4,6,7,8,9,10$, is defined in following proposition.
\begin{proposition}\label{proR}
	Take $\eta=-1$ as example,  $R_j$: $\bar{\Omega}_j\to C$, $j=1,3,4,6,7,8,9,10$ have boundary values as follow:
	\begin{align}
	&R_1(z)=\Bigg\{\begin{array}{ll}
	r(z)T(z)^{-2} & z>z_0,\\
	r(z_0)T_0(z_0)^{-2}(\eta(z-z_0))^{-2i\eta k(z_0)}(1-X_\mathcal{Z}(z))  &z\in \Sigma_1,\\
	\end{array} \\
	&R_3(z)=\Bigg\{\begin{array}{ll}
	\dfrac{\bar{r}(z_0)T_0(z_0)^2}{1+|r(z_0)|^2}(\eta(z-z_0))^{2i\eta k(z_0)}(1-X_\mathcal{Z}(z))  &z\in \Sigma_2,\\
	\dfrac{\bar{r}(z)T_+(z)^2}{1+|r(z)|^2} & 0<z<z_0,\\
	\end{array} \\
	&R_4(z)=\Bigg\{\begin{array}{ll}
	\dfrac{r(z)T_-(z)^{-2}}{1+|r(z)|^2} & 0<z<z_0, \\
	\dfrac{r(z_0)T_0(z_0)^{-2}}{1+|r(z_0)|^2}(\eta(z-z_0))^{-2i\eta k(z_0)}(1-X_\mathcal{Z}(z))  &z\in \Sigma_3,\\
	\end{array} \\
	&R_6(z)=\Bigg\{\begin{array}{ll}
	\bar{r}(z_0)T_0(z_0)^{2}(\eta(z-z_0))^{2i\eta k(z_0)}(1-X_\mathcal{Z}(z))  &z\in \Sigma_4,\\
	\bar{r}(z)T(z)^{2} & z>z_0,\\
	\end{array} \\
	&R_7(z)=\Bigg\{\begin{array}{ll}
	r(z)T(z)^{-2} & z<-z_0,\\
	r(-z_0)T_0(-z_0)^{-2}(\eta(z+z_0))^{-2i\eta k(-z_0)}(1-X_\mathcal{Z}(z))  &z\in \Sigma_6,\\
	\end{array} \\
	&R_8(z)=\Bigg\{\begin{array}{ll}
	\dfrac{\bar{r}(-z_0)T_0(-z_0)^2}{1+|r(-z_0)|^2}(\eta(z+z_0))^{2i\eta k(-z_0)}(1-X_\mathcal{Z}(z))  &z\in \Sigma_5,\\
	\dfrac{\bar{r}(z)T_+(z)^2}{1+|r(z)|^2} & 0>z>-z_0,\\
	\end{array} \\
	&R_9(z)=\Bigg\{\begin{array}{ll}
	\dfrac{r(z)T_-(z)^{-2}}{1+|r(z)|^2} & 0>z>-z_0, \\
	\dfrac{r(-z_0)T_0(-z_0)^{-2}}{1+|r(-z_0)|^2}(\eta(z+z_0))^{-2i\eta k(-z_0)}(1-X_\mathcal{Z}(z))  &z\in \Sigma_8,\\
	\end{array} \\
	&R_{10}(z)=\Bigg\{\begin{array}{ll}
	\bar{r}(-z_0)T_0(-z_0)^{-2}(\eta(z+z_0))^{2i\eta k(-z_0)}(1-X_\mathcal{Z}(z))  &z\in \Sigma_7,\\
	\bar{r}(z)T(z)^{2} & z<-z_0.\\
	\end{array}
	\end{align}	
	And in the case of $\eta=-1$,  $R_j$ is defined follow the in reverse order.   $R_j$  have following property:
	for $j=1,3,4,6,$
	\begin{align}
	&|R_j(z)|\lesssim\sin^2(\arg(z-z_0))+\langle \text{Re}(z)\rangle^{-1/2},\label{R}\\
	&|\bar{\partial}R_j(z)|\lesssim|\bar{\partial}X_Z(z)|+|p_j'(\text{Re}z)|+|z-z_0|^{-1/2},\label{dbarRj}
	\end{align}
	and for $j=7,8,9,10,$
	\begin{align}
	&|R_j(z)|\lesssim \sin^2(\arg(z+z_0))+\langle \text{Re}(z)\rangle^{-1/2},\label{Rk}\\
	&| \bar{\partial}R_j(z)|\lesssim|\bar{\partial}X_Z(z)|+|p_j'(\text{Re}z)|+|z+z_0|^{-1/2},\label{dbarRk}
	\end{align}
	where
	\begin{align}
	&p_1(z)=p_7(z)=r(z),\hspace{2.5cm}p_3(z)=p_8(z)=\dfrac{r(z)}{1+|r(z)|^2},\\
	&p_4(z)=p_9(z)=\dfrac{r(z)}{1+|r(z)|^2},\hspace{1.5cm}p_6(z)=p_{10}(z)=\bar{r}(z).
	\end{align}
	And
	\begin{equation}
	\bar{\partial}R_j(z)=0,\hspace{0.5cm}\text{if } z\in \Omega_2\cup\Omega_5\text{ or }\text{dist}(z,\mathcal{Z}\cup \bar{\mathcal{Z}})<\mu/3.
	\end{equation}
\end{proposition}
The proof of above proposition is similar to  that  in \cite{fNLS,YYL}.  In addition, from the definition of $k(z)$ in (\ref{ks}) and the symmetry of $r(z)$, we have that $k(z_0)=k(-z_0)$.

We now  use $R^{(2)}$ to define the transformation (\ref{transm2}),  which satisfies the following mixed $\bar{\partial}$-RH problem.

\noindent \textbf{RHP4}. Find a matrix valued function  $ M^{(2)}(z;y,t)$ with following properties:

$\blacktriangleright$ Analyticity:  $M^{(2)}(z;y,t)$ is continuous in $\mathbb{C}$,  sectionally continuous first partial derivatives in
$\mathbb{C}\setminus (\Sigma^{(2)}\cup \mathcal{Z}\cup \bar{\mathcal{Z}})$  and meromorphic in $\Omega_2\cup\Omega_5$;

$\blacktriangleright$ Symmetry: $\overline{M^{(2)}(\bar{z})}$=$M^{(2)}(-z)$=$\sigma_2M^{(2)}(z)\sigma_2$;

$\blacktriangleright$ Asymptotic behaviours:
\begin{align}
&M^{(2)}(z;y,t)= I+\mathcal{O}(z^{-1}),\hspace{0.5cm}z \rightarrow \infty;\label{asymbehv5}
\end{align}

$\blacktriangleright$ Jump condition: $M^{(2)}$ has continuous boundary values $M^{(2)}_\pm$ on $\Sigma^{(2)}$ and
\begin{equation}
M^{(2)}_+(z;y,t)=M^{(2)}_-(z;y,t)V^{(2)}(z),\hspace{0.5cm}z \in \Sigma^{(2)},\label{jump4}
\end{equation}
where take $\eta=-1$ as an example, we have
\begin{equation}
V^{(2)}(z)=\left\{\begin{array}{llllllll}
\left(\begin{array}{cc}
1 & R_1(z)e^{2it\theta}\\
0 & 1
\end{array}\right), & z\in \Sigma_1\cup\Sigma_9,\\
\\
\left(\begin{array}{cc}
1 & 0\\
R_3(z)e^{-2it\theta} & 1
\end{array}\right),  &z\in \Sigma_2,\\
\\
\left(\begin{array}{cc}
1 & R_4(z)e^{2it\theta}\\
0 & 1
\end{array}\right),  &z\in \Sigma_3,\\
\\
\left(\begin{array}{cc}
1 & 0\\
R_6(z)e^{-2it\theta} & 1
\end{array}\right),  &z\in \Sigma_4\cup\Sigma_{12},\\
\\
\left(\begin{array}{cc}
1 & R_7(z)e^{2it\theta}\\
0 & 1
\end{array}\right),  &z\in \Sigma_6\cup\Sigma_{10},\\
\\
\left(\begin{array}{cc}
1 & 0\\
R_8(z)e^{-2it\theta} & 1
\end{array}\right),  &z\in \Sigma_5,\\
\\
\left(\begin{array}{cc}
1 & R_9(z)e^{2it\theta}\\
0 & 1
\end{array}\right),  &z\in \Sigma_8,\\
\\
\left(\begin{array}{cc}
1 & 0\\
R_{10}(z)e^{-2it\theta} & 1
\end{array}\right),  &z\in \Sigma_7\cup\Sigma_{11};\\
\end{array}\right.\label{V2}
\end{equation}

$\blacktriangleright$ $\bar{\partial}$-Derivative: For $\mathbb{C}\setminus (\Sigma^{(2)}\cup \mathcal{Z}\cup \bar{\mathcal{Z}})$
we have
\begin{align}
\bar{\partial}M^{(2)}=M^{(1)}\bar{\partial}R^{(2)},
\end{align}
where
\begin{equation}
\bar{\partial}R^{(2)}=\left\{\begin{array}{lll}
\left(\begin{array}{cc}
0 & (-1)^{m_j}\bar{\partial}R_j(z)e^{2it\theta}\\
0 & 0
\end{array}\right), & z\in \Omega_j,j=1,4,7,9,\\
\\
\left(\begin{array}{cc}
0 & 0\\
(-1)^{m_j}\bar{\partial}R_j(z)e^{-2it\theta} & 0
\end{array}\right),  &z\in \Omega_j,j=3,6,8,10,\\
\\
0  &z\in \Omega_2\cup\Omega_5;\\
\end{array}\right. \label{DBARR2}
\end{equation}

$\blacktriangleright$ Residue conditions: $M^{(2)}$ has simple poles at each point in $ \mathcal{Z}\cup \bar{\mathcal{Z}}$ with:

\hspace{1cm} 1.   When $n\in \Delta^+_{z_0,\eta}$,
\begin{align}
&\res_{z=z_n}M^{(2)}(z)=\lim_{z\to z_n}M^{(2)}(z)\left(\begin{array}{cc}
0 & c_nT(z_n)^{-2}e^{-2i\theta_n t}\\
0 & 0
\end{array}\right),\\
&\res_{z=\bar{z}_n}M^{(2)}(z)=\lim_{z\to \bar{z}_n}M^{(2)}(z)\left(\begin{array}{cc}
0 & 0\\
\bar{c}_nT(\bar{z}_n)^2e^{2i\bar{\theta}_n t} & 0
\end{array}\right);
\end{align}

\hspace{1cm} 2.   When $n\in \Delta^-_{z_0,\eta}$,
\begin{align}
&\res_{z=z_n}M^{(2)}(z)=\lim_{z\to z_n}M^{(2)}(z)\left(\begin{array}{cc}
0 & 0\\
c_n(1/T)'(z_n)^{-2}e^{2i\theta_n t}& 0
\end{array}\right),\\
&\res_{z=\bar{z}_n}M^{(2)}(z)=\lim_{z\to \bar{z}_n}M^{(2)}(z)\left(\begin{array}{cc}
0 & \bar{c}_nT'(\bar{z}_k)^{-2}e^{-2i\bar{\theta}_n t}\\
0 & 0
\end{array}\right).
\end{align}

\section{ Decomposition of the mixed $\bar{\partial}$-RH problem }\label{sec6}
\quad To solve RHP4,  we decompose it into a model   RH  Problem  for $M^{RHP}(z;y,t)$  with $\bar\partial R^{(2)}=0$   and a pure $\bar{\partial}$-Problem with $\bar\partial R^{(2)}\not=0$.
For the first step, we establish  a   RH problem  for the  $M^{RHP}(z;y,t)$   as follows.

\noindent\textbf{RHP5}. Find a matrix-valued function  $  M^{RHP}(z;y,t)$ with following properties:

$\blacktriangleright$ Analyticity: $M^{RHP}(z;y,t)$ is analytical  in $\mathbb{C}\setminus (\Sigma^{(2)}\cup \mathcal{Z}\cup \bar{\mathcal{Z}})$;

$\blacktriangleright$ Jump condition: $M^{RHP}$ has continuous boundary values $M^{RHP}_\pm$ on $R$ and
\begin{equation}
M^{RHP}_+(z;y,t)=M^{RHP}_-(z;y,t)V^{(2)}(z),\hspace{0.5cm}z \in R;\label{jump5}
\end{equation}

$\blacktriangleright$ Symmetry: $\overline{M^{RHP}(\bar{z})}$=$M^{RHP}(-z)$=$\sigma_2M^{RHP}(z)\sigma_2$;

$\blacktriangleright$ $\bar{\partial}$-Derivative:  $\bar{\partial}R^{(2)}=0$, for $ z\in \mathbb{C}$;

$\blacktriangleright$ Asymptotic behaviours:
\begin{align}
&M^{RHP}(z;y,t)= I+\mathcal{O}(z^{-1}),\hspace{0.5cm}z \rightarrow \infty;\label{asymbehv6}
\end{align}

$\blacktriangleright$ Residue conditions: $M^{RHP}$ has simple poles at each point in $ \mathcal{Z}\cup \bar{\mathcal{Z}}$ with:

\hspace{1cm} 1.   When $n\in \Delta^+_{z_0,\eta}$,
\begin{align}
&\res_{z=z_n}M^{RHP}(z)=\lim_{z\to z_n}M^{RHP}(z)\left(\begin{array}{cc}
0 & c_nT(z_n)^{-2}e^{-2i\theta_n t}\\
0 & 0
\end{array}\right),\label{6.3}\\
&\res_{z=\bar{z}_n}M^{RHP}(z)=\lim_{z\to \bar{z}_n}M^{RHP}(z)\left(\begin{array}{cc}
0 & 0\\
\bar{c}_nT(\bar{z}_n)^2e^{2i\bar{\theta}_n t} & 0
\end{array}\right).
\end{align}

\hspace{1cm} 2.   When $n\in \Delta^-_{z_0,\eta}$,
\begin{align}
&\res_{z=z_n}M^{RHP}(z)=\lim_{z\to z_n}M^{RHP}(z)\left(\begin{array}{cc}
0 & 0\\
c_n(1/T)'(z_n)^{-2}e^{2i\theta_n t}& 0
\end{array}\right),\\
&\res_{z=\bar{z}_n}M^{RHP}(z)=\lim_{z\to \bar{z}_n}M^{RHP}(z)\left(\begin{array}{cc}
0 & \bar{c}_nT'(\bar{z}_k)^{-2}e^{-2i\bar{\theta}_n t}\\
0 & 0
\end{array}\right).\label{6.6}
\end{align}
The existence  and asymptotic  of  $M^{RHP}(z)$  will shown in   section 8.

We now use $M^{RHP}(z)$ to construct  a new matrix function
\begin{equation}
M^{(3)}(z)=M^{(2)}(z)M^{RHP}(z)^{-1}.\label{transm3}
\end{equation}
which   removes   analytical component  $M^{RHP}$    to get  a  pure $\bar{\partial}$-problem.

\noindent\textbf{RHP6}. Find a matrix-valued function  $ M^{(3)}(z;y,t)$ with following properties:

$\blacktriangleright$ Analyticity: $M^{(3)}(z;y,t)$ is continuous with sectionally continuous first partial derivatives in
 $\mathbb{C}\setminus (\Sigma^{(2)}\cup \mathcal{Z}\cup \bar{\mathcal{Z}})$ and meromorphic in $\Omega_2\cup\Omega_5$.

$\blacktriangleright$ Symmetry: $\overline{M^{(3)}(\bar{z})}$=$M^{(3)}(-z)$=$\sigma_2M^{(3)}(z)\sigma_2$;

$\blacktriangleright$ Asymptotic behavior:
\begin{align}
&M^{(3)}(z;y,t) \sim I+\mathcal{O}(z^{-1}),\hspace{0.5cm}z \rightarrow \infty;\label{asymbehv7}
\end{align}

$\blacktriangleright$ $\bar{\partial}$-Derivative: For $\mathbb{C}\setminus (\Sigma^{(2)}\cup \mathcal{Z}\cup \bar{\mathcal{Z}})$ we have $\bar{\partial}M^{(3)}=M^{(3)}W^{(3)}$,
\begin{equation}
W^{(3)}=M^{RHP}(z)\bar{\partial}R^{(2)}M^{RHP}(z)^{-1}.
\end{equation}

\begin{proof}
	By using  properties  of  the   solutions   $M^{(2)}$ and $M^{RHP}$  for  RHP4  and  RHP5,
 the analyticity and asymptotics are obtained   immediately.  Since $M^{(2)}$ and $M^{RHP}$ have same jump matrix, we have
	\begin{align*}
	M_-^{(3)}(z)^{-1}M_+^{(3)}(z)&=M_-^{(2)}(z)^{-1}M_-^{RHP}(z)M_+^{RHP}(z)^{-1}M_+^{(2)}(z)\\
	&=M_-^{(2)}(z)^{-1}V^{(2)}(z)^{-1}M_+^{(2)}(z)=I,
	\end{align*}
	which means $ M^{(3)}$ has no jumps and is everywhere continuous.  We also can show  that $ M^{(3)}$ has no pole. For
 $\lambda \in \mathcal{Z}\cup \bar{\mathcal{Z}}$,  let $\mathcal{N}$ denote the  nilpotent matrix which appears in the left side of the
corresponding residue condition of RHP4  and  RHP5,
 we have the Laurent expansions in $z-\lambda$
	\begin{align}
&M^{(2)}(z)=a(\lambda) \left[ \dfrac{\mathcal{N}}{z-\lambda}+I\right] +\mathcal{O}(z-\lambda),\nonumber\\
&	M^{RHP}(z)=A(\lambda) \left[ \dfrac{N\mathcal{}}{z-\lambda}+I\right] +\mathcal{O}(z-\lambda),\nonumber
\end{align}
	where $a(\lambda)$ and $A(\lambda)$ are the constant row vector and matrix in their respective expansions.
Then  from $M^{RHP}(z)^{-1}=\sigma_2M^{RHP}(z)^T\sigma_2$,  we have
	\begin{align}
	M^{(3)}(z)&=\left\lbrace a(\lambda) \left[ \dfrac{N}{z-\lambda}+I\right]\right\rbrace \left\lbrace\left[ \dfrac{-N}{z-\lambda}+I\right]\sigma_2A(\lambda)^T\sigma_2\right\rbrace + \mathcal{O}(z-\lambda)\nonumber\\
	&=\mathcal{O}(1),
	\end{align}
	which  implies that  $ M^{(3)}$ has removable singularities at $\lambda$.
 And the $\bar{\partial}$-derivative of  $ M^{(3)}$ come  from    $ M^{(3)}$  due to   analyticity of $M^{RHP}$.
\end{proof}

  We construct the solution $M^{RHP}$ of the RHP5 in the following form
\begin{equation}
M^{RHP}=\left\{\begin{array}{lll}
E(z)M^{(out)}(z) & z\notin U_{z_0},\\
E(z)M^{(z_0)}(z)  &z\in U_{z_0,}\\
E(z) M^{(-z_0)}(z) &z\in U_{-z_0},\\
\end{array}\right. \label{transm4}
\end{equation}
where $ U_{\pm z_0}$ are  the neighborhoods of $\pm z_0$, respectively
\begin{equation}
U_{\pm z_0}=\left\lbrace z:|z\mp z_0|\leq \min\left\lbrace\frac{z_0}{2} ,\mu/3\right\rbrace \triangleq\varepsilon\right\rbrace .\label{Uz0}
\end{equation}
This implies that $M^{RHP}$   and  $M^{(\pm z_0)}$ have no poles in  $ U_{\pm z_0}$, since $ {\rm dist}(\mathcal{Z},\mathbb{R})>\mu$.
This decomposition  splits   $M^{RHP}$  into two parts:  $M^{(out)}$ solves a model  RHP obtained by ignoring the jump conditions of  RHP5,
 which will be solved   in next Section \ref{sec7};   While  $M^{(\pm z_0)}$,  whose solution can be  approximated  with   parabolic cylinder functions
 if we let  $M^{(\pm z_0)}$   exactly match to the  $M^{(2)}$  and a  parabolic cylinder model  in $ U_{\pm z_0}$,  these results will given  in Section \ref{sec8}.
  And $E(z)$ is a error function, which is a solution of a small-norm RH problem and we discuss it in Section \ref{sec9}.

And  from  the  RHP5, whose jump matrix admits the  following extimates.
\begin{proposition}\label{pro3v2}
	For the jump matrix $ V^{(2)}(z)$, we have the following  estimate
	\begin{align}
	&\parallel V^{(2)}-I\parallel_{L^\infty(\Sigma^{(2)}_\pm\cap U_{\pm z_0})}=\mathcal{\mathcal{O}}(e^{- \frac{\sqrt{2}|t|}{4}|z\mp z_0|}(z_0^{-2}-|z|^{-2}) ),\label{7.1}\\
	&\parallel V^{(2)}-I\parallel_{L^\infty(\Sigma^{(2)}_0)}=\mathcal{\mathcal{O}}(e^{-\frac{|t|}{4z_0} } ),\label{7.2}
	\end{align}
	where  the contours are defined by
\begin{align}
&\Sigma^{(2)}_+=\Sigma_1\cup\Sigma_2\cup\Sigma_3\cup\Sigma_4, \ \ \Sigma^{(2)}_-=\Sigma_5\cup\Sigma_6\cup\Sigma_7\cup\Sigma_8,\nonumber\\
&\Sigma^{(2)}_0=\Sigma_{10}\cup\Sigma_{11}\cup\Sigma_{12}\cup\Sigma_9.\nonumber
	\end{align}
\end{proposition}
\begin{proof}
	We   prove (\ref{7.2})  for  the case when  $\eta=+1$ and  $z\in\Sigma_{9}$,    other cases can be shown in a  similar  way.
	By using  definition of $V^{(2)}$ and (\ref{R}),  we have
	\begin{align}
	\parallel V^{(2)}-I\parallel_{L^\infty(\Sigma_9)}\leq \parallel R_1e^{2it\theta}\parallel_{L^\infty(\Sigma_9)}. \label{poope}
	\end{align}
	Note that $|z|<\sqrt{2}z_0/2$ for $z\in \Sigma_{9}$, together with (\ref{theta}),  we find that
	\begin{align}
	|e^{2it\theta}|= e^{-2t\text{Im}z\left(\xi+\frac{1}{4|z|^2}\right) }\leq e^{-\frac{t}{4z_0} } \to 0,\hspace{0.3cm} \text{as }t\to + \infty,\nonumber
	\end{align}
	which together with (\ref{poope}) gives (\ref{7.2}). And the calculation of $\Sigma^{(2)}_\pm$ is similar.
\end{proof}

This proposition means that the jump matrix $V^{(2)}$    uniformly goes to  $I$  on  both   $\Sigma^{(2)}_\pm\cap U_{\pm z_0}$ and $\Sigma^{(2)}_0$,
so outside the $U_{z_0}\cup U_{-z_0}$ there is only exponentially small error (in t) by completely ignoring the jump condition of  $M^{RHP}$. And note that unlike the neighborhood of $\pm z_0$, $V^{(2)}\to I$ as $z\to0$, it has  uniformly property. So we doesn't need to consider the neighborhood of $z=0$ alone.

\section{ Outer   model  RH problem } \label{sec7}
\quad In this section, we build a outer model RH  problem and  show that  its solution can  approximated  with  a finite sum  of soliton solutions.
Note that from the reconstruct formula (\ref{recons u}), we only need the property of $M^{(out)}$ as $z\to 0$. We can introduce following outer model problem.

\noindent\textbf{RHP7.}  Find a matrix-valued function  $ M^{(out)}(z;y,t)$ with following properties:

$\blacktriangleright$ Analyticity: $M^{(out)}(z;y,t)$ is analytical  in $\mathbb{C}\setminus (\Sigma^{(2)}\cup \mathcal{Z}\cup \bar{\mathcal{Z}})$;

$\blacktriangleright$ Symmetry: $\overline{M^{(out)}(\bar{z})}$=$M^{(out)}(-z)$=$\sigma_2M^{(out)}(z)\sigma_2$;

$\blacktriangleright$ Asymptotic behaviours:
\begin{align}
&M^{(out)}(z;y,t) \sim I+\mathcal{O}(z^{-1}),\hspace{0.5cm}z \rightarrow \infty;
\end{align}

$\blacktriangleright$ Residue conditions: $M^{(out)}$ has simple poles at each point in $ \mathcal{Z}\cup \bar{\mathcal{Z}}$ satisfying the
same residue relations  (\ref{6.3})-(\ref{6.6}) with $M^{RHP}(z)$.

Before showing the existence and uniqueness of solution of the above RHP7,  we  first consider the reflectionless case
of the RHP1. In this case,  $M$ has  no  contour,  the RHP1   reduces to the following RH problem.

\noindent \textbf{RHP8.}  Given discrete data $\sigma_d=\left\lbrace(z_k,c_k) \right\rbrace _{k=1}^N$, and $\mathcal{Z}=\left\lbrace z_k\right\rbrace _{k=1}^N $.
 Find a matrix-valued function  $  m(z;y,t|\sigma_d)$ with following properties:
	
	$\blacktriangleright$ Analyticity: $m(z;y,t|\sigma_d)$ is analytical  in $\mathbb{C}\setminus (\Sigma^{(2)}\cup \mathcal{Z}\cup \bar{\mathcal{Z}})$;
	
	$\blacktriangleright$ Symmetry: $\overline{m(\bar{z};y,t|\sigma_d)}$=$m(-z;y,t|\sigma_d)$=$\sigma_2m(z;y,t|\sigma_d)\sigma_2$;
	
	$\blacktriangleright$ Asymptotic behaviours:
	\begin{align}
	&m(z;y,t|\sigma_d) \sim I+\mathcal{O}(z^{-1}),\hspace{0.5cm}z \rightarrow \infty;
	\end{align}

	$\blacktriangleright$ Residue conditions: $m(z;y,t|\sigma_d)$ has simple poles at each point in $ \mathcal{Z}\cup \bar{\mathcal{Z}}$ satisfying
	\begin{align}
	&\res_{z=z_n}m(z;y,t|\sigma_d)=\lim_{z\to z_n}m(z;y,t|\sigma_d)\tau_k,\\
	&\res_{z=\bar{z}_n}m(z;y,t|\sigma_d)=\lim_{z\to \bar{z}_n}m(z;y,t|\sigma_d)\hat{\tau}_k,
	\end{align}
	where $\tau_k$ is a nilpotent matrix satisfies
	\begin{align}
	\tau_k=\left(\begin{array}{cc}
	0 & \gamma_k\\
	0 & 0
	\end{array}\right),\hspace{0.5cm}\hat{\tau}_k=\sigma_2\overline{\tau_k}\sigma_2,\hspace{0.5cm} \gamma_k=c_ke^{-2it\theta_k}.
	\end{align}
	Moreover, the solution satisfies
	\begin{equation}
	\parallel m(z;y,t|\sigma_d)^{-1}\parallel_{L^\infty(C\setminus(\mathcal{Z}\cup \bar{\mathcal{Z}}))}\lesssim 1. \label{normm}
	\end{equation}

\begin{proposition}\label{unim}	 The RHP8   exists an  unique solution.
\end{proposition}

\begin{proof}
	The uniqueness of solution follows from the Liouville's theorem. The symmetries  of $m(z;y,t|\sigma_d)$  means that
it  admits a partial fraction expansion of following form
	\begin{equation}
	m(z;y,t|\sigma_d)=I+\sum_{k=1}^{N}\left[ \frac{1}{z-z_k}\left(\begin{array}{cc}
	0 & \nu_k(y,t)\\
	0 & \varsigma_k(y,t)
	\end{array}\right)+\frac{1}{z-\bar{z}_k}\left(\begin{array}{cc}
	\overline{\varsigma_k(y,t)} & 0\\
	-\overline{\nu_k(y,t)} & 0
	\end{array}\right)\right] .\label{EXPM}
	\end{equation}
	By using a  similar way  to Appendix B  in \cite{NLSB},    we can show   the existence of the solution for the RHP8.	
	Since det($m(z;y,t|\sigma_d)$)=1,  $\parallel m(z;y,t|\sigma_d)\parallel_{L^\infty(C\setminus(\mathcal{Z}\cup \bar{\mathcal{Z}}))}$ is bounded.
 And from (\ref{EXPM}),  we simply obtain (\ref{normm}).
\end{proof}

In  reflectionless case, the transmission coefficient admits   following trace formula
\begin{equation}
a(z)=\prod_{k=1}^N\dfrac{z-z_k}{z-\bar{z}_k},\label{az}
\end{equation}
whose poles can be split into two parts.   Let $\triangle \subseteqq\left\lbrace 1,2,...,N\right\rbrace $, and define
\begin{equation}
a_\triangle(z)=\prod_{k\in\triangle }\dfrac{z-z_k}{z-\bar{z}_k},\nonumber
\end{equation}
we  make  a  renormalization transformation
\begin{equation}
m^\triangle(z|D) =m(z|\sigma_d)a^\triangle(z)^{-\sigma_3},\label{renormalization}
\end{equation}
where the scattering data are given by
\begin{align}
D=\left\lbrace(z_k,c_k') \right\rbrace _{k=1}^N, \ \ c_k'=c_ka^\triangle(z)^2.\label{delta}
\end{align}
It is easy to see that the  transformation (\ref{renormalization})   splits the poles between the columns of $m^\triangle(z|D)$
according to the choice of $\triangle$, and it satisfies the following modified discrete RH problem.

\noindent\textbf{RHP9.}   Given discrete data (\ref{delta}),  find a matrix-valued function  $ m^\triangle(z;y,t|D)$ with following properties:

$\blacktriangleright$ Analyticity: $m^\triangle(z;y,t|D)$ is analytical  in $\mathbb{C}\setminus (\Sigma^{(2)}\cup \mathcal{Z}\cup \bar{\mathcal{Z}})$;

$\blacktriangleright$ Symmetry:  $m^\triangle(z;y,t|D)=\sigma_2\overline{m^\triangle(\bar{z};y,t|D)}\sigma_2=\overline{m^\triangle(-\bar{z};y,t|D)}$;

$\blacktriangleright$ Asymptotic behaviours:
\begin{align}
&m^\triangle(z;y,t|D) \sim I+\mathcal{O}(z^{-1}),\hspace{0.5cm}z \rightarrow \infty;
\end{align}

$\blacktriangleright$ Residue conditions: $m(z;y,t|\sigma_d)$ has simple poles at each point in $ \mathcal{Z}\cup \bar{\mathcal{Z}}$ satisfying
\begin{align}
&\res_{z=z_n}m^\triangle(z;y,t|D)=\lim_{z\to z_n}m^\triangle(z;y,t|D)\tau^\triangle_k,\\
&\res_{z=\bar{z}_n}m^\triangle(z;y,t|D)=\lim_{z\to \bar{z}_n}m^\triangle(z;y,t|D)\hat{\tau}_k^\triangle,
\end{align}
where $\tau_k$ is a nilpotent matrix satisfies
\begin{align}
&\tau_k^\triangle=\left\{\begin{array}{ll}
\left(\begin{array}{cc}
0 & \gamma_ka^\triangle(z)^2\\
0 & 0
\end{array}\right) & k\notin \triangle\\
\left(\begin{array}{cc}
0 & 0\\
\gamma_k^{-1}a'^\triangle(z)^{-2} & 0
\end{array}\right)  & k\in \triangle\\
\end{array}\right.,\hspace{0.5cm}\hat{\tau}_k^\triangle=\sigma_2\overline{\tau}_k^\triangle\sigma_2^{-1},\nonumber\\
&\gamma_k=c_ke^{-2it\theta_k}.\label{tau}
\end{align}
Since (\ref{renormalization}) is a explicit transformation of $m(z;y,t|\sigma_d)$,  by  Proposition  \ref{unim},
 we obtain the existence and uniqueness of the solution of the  RHP9.

In the  RHP9, take $\Delta =\Delta^-_{z_0,\eta}$ and replace the scattering data  $D$ with  scattering data
\begin{align}
\tilde{D}=\left\lbrace(z_k,\tilde{c}_k) \right\rbrace _{k=1}^N, \ \ \tilde{c}_k=c_k\delta(z_k)^2,\label{out}
\end{align}
then we have
\begin{corollary}
 	 There exists and   unique solution  for the RHP7, moreover,
 	\begin{equation*}
 	M^{(out)}(z;y,t)=m^{\Delta^-_{z_0,\eta}}(z;y,t|\tilde{D}),
 	\end{equation*}
 	where scattering data  $\tilde{D}$ is given by (\ref{out}).
\end{corollary}

   If $u_{sol}(y,t)=u_{sol}(y,t;D)$ denotes the $N$-soliton solution of the SP equation (\ref{sp}) encoded by the RHP8,
  by using (\ref{renormalization}),  we also have the reconstruction formula
\begin{equation}
u_{sol}(x,t|D)=u_{sol}(y(x,t),t|D)=\lim_{z\to 0}\dfrac{\left( m^\triangle(0;y,t|D)^{-1}m^\triangle(z;y,t|D) \right)_{12} }{iz},
\end{equation}
which show that each normalization encodes $u_{sol}(y,t)$ in the same way. If we choosing $\triangle$ appropriately, the asymptotic limits $|t|\to\infty$ with $\xi=y/t$ bounded are under better asymptotic control. Then we consider the long-time behavior of soliton solutions.

Give pairs points $y_1\leq y_2\in \mathbb{R}$ and velocities $v_1\leq v_2 \in \mathbb{R}^-$,  we  define a  cone
\begin{equation}
C(y_1,y_2,v_1,v_2)=\left\lbrace (y,t)\in R^2|y=y_0+vt, \ y_0\in[y_1,y_2]\text{, }v\in[v_1,v_2]\right\rbrace.\label{coneC}
\end{equation}
and denote
\begin{align}
&I=\left\lbrace z:\ -\frac{1}{4v_1}<|z|^2<-\frac{1}{4v_2} \right\rbrace,\nonumber \\
&\mathcal{Z}(I)=\left\lbrace z_k\in \mathcal{Z}: \ z_k\in I\right\rbrace ,\hspace{2.1cm}N(I)=|\mathcal{Z}(I)|,\nonumber \\
&\mathcal{Z}^-(I)=\left\lbrace z_k\in \mathcal{Z}: \ |z|^2>-\frac{1}{4v_2}\right\rbrace,\hspace{0.5cm}\mathcal{Z}^+(I)=\left\lbrace z_k\in \mathcal{Z}: \ |z|^2<-\frac{1}{4v_1}\right\rbrace\nonumber,\\
&c_k(I)=c_k\prod_{\text{Re}z_n\in I_\eta\setminus I}\left( \frac{z_k-z_n}{z_k-\bar{z}_n}\right) ^2\exp\left[-\frac{1}{\pi i}\int_{I_\eta}\frac{\log[1+|r(\zeta)|^2]}{\zeta-z}d\zeta\right].\label{dataI}
\end{align}
We can show the  following lemma.

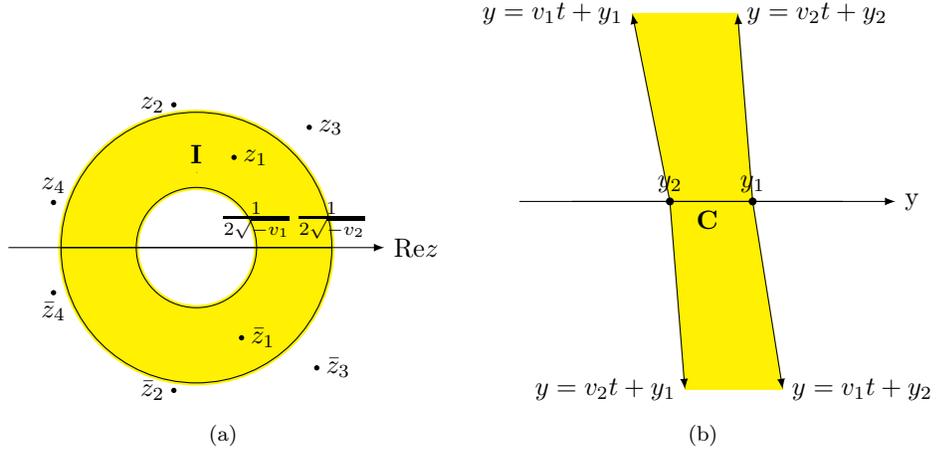
\begin{figure}[H]
	\centering
	\centering
	\subfigure[]{
		\begin{tikzpicture}[node distance=2cm]
		\filldraw[yellow,line width=2] (0.8,0)--(1.8,0)arc(0:180:1.8)--(-1.8,0)--(-0.8,0)arc(180:0:0.8)--(0.8,0);
		\filldraw[yellow,line width=2] (0.8,0)--(1.8,0)arc(360:180:1.8)--(-1.8,0)--(-0.8,0)arc(180:360:0.8)--(0.8,0);
		\draw[-latex](-2.5,0)--(2.5,0)node[right]{Re$z$};
		\draw(0,0) circle (0.8);
		\draw(0,0) circle (1.8);
		\draw[-](0,0)--(-0.8,0);
		\draw[-](-0.8,0)--(-1.8,0);
		\draw[-](0,0)--(0.8,0)node[above] {$\frac{1}{2\sqrt{-v_1}}$};
		\draw[-](0.8,0)--(1.8,0)node[above] {$\frac{1}{2\sqrt{-v_2}}$};
		\coordinate (A) at (0.5,1.2);
		\coordinate (B) at (0.6,-1.2);
		\coordinate (C) at (1.5,1.6);
		\coordinate (D) at (1.6,-1.6);
		\coordinate (E) at (-1.9,0.6);
		\coordinate (F) at (-1.9,-0.6);
		\coordinate (G) at (-0.3,1.9);
		\coordinate (H) at (-0.3,-1.9);
		\coordinate (I) at (0,1);
		\fill (A) circle (1pt) node[right] {$z_1$};
		\fill (B) circle (1pt) node[right] {$\bar{z}_1$};
		\fill (G) circle (1pt) node[left] {$z_2$};
		\fill (H) circle (1pt) node[left] {$\bar{z}_2$};
		\fill (C) circle (1pt) node[right] {$z_3$};
		\fill (D) circle (1pt) node[right] {$\bar{z}_3$};
		\fill (E) circle (1pt) node[above] {$z_4$};
		\fill (F) circle (1pt) node[below] {$\bar{z}_4$};
		\fill (I) circle (0pt) node[above] {$\textbf{I}$};
		\label{figC}
		\end{tikzpicture}
	}
	\subfigure[]{
		\begin{tikzpicture}[node distance=2cm]
		\draw[yellow, fill=yellow] (0.6,0)--(0.4,2.5)--(-1.0,2.5)--(-0.5,0)--(-0.3,-2.5)--(1.0,-2.5)--(0.6,0);
		\draw[-latex](-2.5,0)--(2.5,0)node[right]{y};
		\draw[-latex](0.6,0)--(0.4,2.5)node[right]{$y=v_2t+y_2$};
		\draw[-latex](-0.5,0)--(-0.3,-2.5)node[left]{$y=v_2t+y_1$};
		\draw[-latex](-0.5,0)--(-1.0,2.5)node[left]{$y=v_1t+y_1$};
		\draw[-latex](0.6,0)--(1.0,-2.5)node[right]{$y=v_1t+y_2$};
		\draw[-](0,0)--(-0.8,0);
		\draw[-](-0.8,0)--(-1.8,0);
		\draw[-](0,0)--(0.8,0);
		\draw[-](0.8,0)--(1.8,0);
		\coordinate (A) at (0.6,0);
		\coordinate (B) at (-0.5,0);
		\coordinate (I) at (0,0);
		\fill (A) circle (1.5pt) node[above] {$y_1$};
		\fill (B) circle (1.5pt) node[above] {$y_2$};
		\fill (I) circle (0pt) node[below] {$\textbf{C}$};
		\label{figzero}
		\end{tikzpicture}
	}
	\caption{(a) In the example here, the original data has four pairs zero points of discrete spectrum, but insider the cone C only three pairs
 points with $\mathcal{Z}(I)=\left\lbrace z_1 \right\rbrace$; (b) The cone $C(y_1,y_2,v_1,v_2)$}
	\label{figC(I)}
\end{figure}

\begin{lemma}
	Fix reflectionless data $D=\left\lbrace (z_k,c_k')\right\rbrace_{k=1}^N $, $D(I)=\left\lbrace (z_k,c_k'(I))|z_k\in \mathcal{Z}(I)\right\rbrace$. Then as $|t|\to\infty$ with $(y,t)\in C(y_1,y_2,v_1,v_2)$, we have
	\begin{equation}
m^{\triangle^-_{z_0,\eta}}(z;y,t|D)=\left( I+\mathcal{O}(e^{-2\mu(I)|t|})\right) m^{\triangle^-_{z_0,\eta}}(z;y,t|D(I)),\label{DI}
	\end{equation}
	where $\mu(I)=\min_{z_K\in \mathcal{Z}\setminus \mathcal{Z}(I)}\left\lbrace \text{\rm Re}(z_k)\frac{-v_2}{|z|^2}(|z|+\frac{1}{2\sqrt{-v_1}})  {\rm dist}(z_k,I)\right\rbrace $.
\end{lemma}
\begin{proof} We denote
$$\triangle^+(I)=\left\lbrace k|\text{Re}(z_k)<-v_2/4\right\rbrace,  \ \
\triangle^-(I)=\left\lbrace k|\text{Re}(z_k)>-v_1/4\right\rbrace.$$
and take $\triangle=\triangle_{z_0,\eta}^-$  with $\eta=\text{sgn}(t)$ in the  RHP9,
then for $z\in \mathcal{Z}\setminus \mathcal{Z}(I)$ and $(y,t)\in C(y_1,y_2,v_1,v_2)$, by using the residue coefficients (\ref{tau}),
   direct calculation shows that
	\begin{align}
& |\gamma_k|=|c_ke^{2\text{Re}(z)y_0}||e^{2t\text{Re}(z)(\frac{1}{4|z|^2}+v_0)}|\nonumber\\
&=|c_ke^{2\text{Re}(z)y_0}| \big|e^{-2t\text{Re}(z)\frac{-v_0}{|z|^2}(|z|+\frac{1}{2\sqrt{-v_0}})(|z|-\frac{1}{2\sqrt{-v_0}})}\big|,\nonumber
	\end{align}
which leads to
	\begin{equation}
	\parallel\tau_k^{\triangle^\pm(I)}\parallel=\mathcal{O}(e^{-2\mu(I)|t|}),\hspace{0.5cm}  t\to \pm\infty.\label{asytau}
	\end{equation}
	Suppose that $D_k$ is a small disks centrad in each $z_k\in \mathcal{Z}\setminus \mathcal{Z}(I)$ with  radius smaller than $\mu$. Denote $\partial D_k$ is the boundary of $D_k$. Then we can  introduce a new transformation which can remove the poles $z_k\in \mathcal{Z}\setminus \mathcal{Z}(I)$ and these residues change to near-identity jumps.
	\begin{equation}\label{7.24}
	 \widetilde{m}^{\triangle^-_{z_0,\eta}}(z;y,t|D)=\left\{\begin{array}{ll}
	m^{\triangle^-_{z_0,\eta}}(z;y,t|D)\left(I-\frac{\tau^{\triangle^\eta(I)}_k}{z-z_k} \right)  & z\in D_k,\\
	m^{\triangle^-_{z_0,\eta}}(z;y,t|D)\left(I-\frac{\hat{\tau}^{\triangle^\eta(I)}_k}{z-\bar{z}_k} \right)  & z\in \bar{D}_k,\\
	m^{\triangle^-_{z_0,\eta}}(z;y,t|D)  & elsewhere.\\
	\end{array}\right.
	\end{equation}
	Comparing with $m^{\triangle^-_{z_0,\eta}}$, the new matrix function $\widetilde{m}^{\triangle^-_{z_0,\eta}}(z;y,t|D)$ has new jump in each $\partial D_k$ which denote by $\widetilde{V}(z)$. Then using (\ref{asytau}),  we have
	\begin{equation}
	\parallel\widetilde{V}(z)-I\parallel_{L^\infty(\widetilde{\Sigma})}=\mathcal{O}(e^{-2\mu(I)|t|}),\hspace{0.5cm}\widetilde{\Sigma}=\cup_{z_k\in \mathcal{Z}\setminus \mathcal{Z}(I)}\left( \partial D_k\cup\partial \bar{D}_k\right).\label{asytV}
	\end{equation}
Since  $ \widetilde{m}^{\triangle^-_{z_0,\eta}}(z;y,t|D) $ has same poles and residue conditions  with  $  {m}^{\triangle^-_{z_0,\eta}}(z;y,t|D(I))$,
then
$$m_0(z)= \widetilde{m}^{\triangle^-_{z_0,\eta}}(z;y,t|D)  {m}^{\triangle^-_{z_0,\eta}}(z;y,t|D(I))^{-1}$$
  has no poles, but  it has jump matrix for $z\in\widetilde{\Sigma}$,
	\begin{equation}
	m_0^+(z)=m_0^-(z)V_{m_0}(z),
	\end{equation}
where the jump matrix $V_{m_0}(z)$ given by
	\begin{equation}
 V_{m_0}(z)=m(z|D(I))\widetilde{V}(z)m(z|D(I))^{-1},
	\end{equation}
which, by using (\ref{asytV}),  also admits the same decaying estimate
$$\parallel V_{m_0}(z)-I\parallel_{L^\infty(\widetilde{\Sigma})}=\parallel\widetilde{V}(z)-I\parallel_{L^\infty(\widetilde{\Sigma})}=\mathcal{O}(e^{-2\mu(I)|t|}), \ \ t\rightarrow \pm\infty.$$

 Then by using  the theory of small norm  RH problem \cite{RN4,RN5},   we find  that
  $m_0(z)$ exists and
  $$m_0(z)=I+\mathcal{O}(e^{-2\mu(I)|t|}), \ \ t\to \pm\infty,$$
which together with (\ref{7.24})  gives  the formula (\ref{DI}).
\end{proof}

Using reconstruction formula  to $m^{\triangle^-_{z_0,\eta}}(z;y,t|D)$,  we  immediately obtain the following  result.
\begin{corollary}\label{usol}
	Let $u_{sol}(y,t;D)$ and $u_{sol}(y,t;D(I))$ denote the $N$-soliton solution of (\ref{sp}) corresponding to discrete scattering data D and D(I) respectively.
 As $|t|\to\infty$ with $(y,t)\in C(y_1,y_2,v_1,v_2)$, we have
	\begin{align}
	&\lim_{z\to 0}\dfrac{\left( m^{\triangle^-_{z_0,\eta}}(0;y,t|D)^{-1}m^{\triangle^-_{z_0,\eta}}(z;y,t|D) \right)_{21} }{iz}=u_{sol}(y(x,t),t;D)\nonumber\\
	&=u_{sol}(y(x,t),t;D(I))+\mathcal{O}(e^{-\mu(I)|t|}).
	\end{align}
By using  (\ref{c+}), we have
	\begin{align}
	c_+(x,t;D)=&\lim_{z\to 0}\dfrac{\left( m^{\triangle^-_{z_0,\eta}}(0;y,t|D)^{-1}m^{\triangle^-_{z_0,\eta}}(z;y,t|D) \right)_{11}-1 }{iz}\nonumber\\
	&=c_+(x,t;D(I))+\mathcal{O}(e^{-\mu(I)|t|}).\label{c+D}
	\end{align}
\end{corollary}
Now  we come back to the outer model and obtain the  following result.
\begin{corollary}
	The RHP7   exists  an  unique solution $M^{(out)}$  	with
	\begin{align}
	M^{(out)}(z)&=m^{\triangle^-_{z_0,\eta}}(z|D^{(out)})\nonumber\\
	&=m^{\triangle^-_{z_0,\eta}}(z;y,t|D(I))\prod_{\text{Re}z_n\in I_{z_0}^\eta\setminus I}\left( \frac{z-z_n}{z-\bar{z}_n}\right)^{-\sigma_3} \delta^{-\sigma_3} +\mathcal{O}(e^{-\mu(I)|t|}),\label{mout}
	\end{align}
	where $D^{(out)}=\left\lbrace z_k,c_k(z_0)\right\rbrace_{k=1}^N $ with
	\begin{equation*}
	c_k(\xi)=c_k\exp\left[-\frac{1}{\pi i}\int_{I_\eta}\frac{\log[1+|r(\zeta)|^2]}{\zeta-z}d\zeta\right].
	\end{equation*}
	Then substitute (\ref{mout})  into (\ref{normm}) we immediately have
	\begin{equation}
	\parallel M^{(out)}(z)^{-1}\parallel_{L^\infty(C\setminus(\mathcal{Z}\cup \bar{\mathcal{Z}}))}\lesssim 1.\label{normmout}
	\end{equation}
	Moreover, we have reconstruction formula
	\begin{equation}
	\lim_{z\to 0}\dfrac{\left( M^{(out)}(0)^{-1}M^{(out)}(z) \right)_{12} }{iz}=u_{sol}(y,t;D^{(out)}),\label{expMout}
	\end{equation}
	where the $u_{sol}(y,t;D^{out})$ is the $N$-soliton solution of (\ref{sp}) corresponding to discrete scattering data $\widetilde{D}$. And
	\begin{equation}
	u_{sol}(y,t;D^{(out)})=u_{sol}(y,t;D(I))+\mathcal{O}(e^{-\mu(I)|t|}),\hspace{0.5cm} \text{for }t\to\pm\infty.\label{asyqout}
	\end{equation}
\end{corollary}

\section{ A local solvable  RH model near phase points} \label{sec8}

\quad From  the Proposition  \ref{pro3v2}, in the neighborhood $U_{\pm z_0}$ of $\pm z_0$, we find that $V^{(2)}-I$ doesn't  have a uniformly small jump for large time,
so we establish  a local model  for function $E(z)$ with   a uniformly small jump.

For soliton-free case when there are no discrete spectrum,   the formula  (\ref{T}) and (\ref{T0}) reduce  to  $T_0(\pm z_0)=\delta(\pm z_0)$.
The   RHP5    exactly   reduces to a  solvable  model for the SP equation \cite{xusp}.

\noindent\textbf{RHP10.}  Let   $\Sigma^{(2)}$ be the same contour  in the  RHP5. Find a matrix-valued   function $M^{sp}(z;\eta)$  such that

$\bullet$ Analyticity: $M^{sp}(z;\eta)$ is analytical  in $\mathbb{C}\setminus \Sigma^{(2)}$;

$\bullet$ Symmetry: $M^{sp}(z;\eta= 1 )=\sigma_2M^{sp}(-z;\eta=-1  )\sigma_2$;

$\bullet$ Asymptotic behaviors:
\begin{align}
&M^{sp}(z;\eta) \sim I+\mathcal{O}(z^{-1}),\hspace{0.5cm}z \rightarrow \infty.
\end{align}

$\bullet$  Jump condition: $M^{sp}(z;\eta  )$ has continuous boundary values $M^{sp}_\pm(z;\eta )$ on $\Sigma^{(2)}$ and
\begin{equation}
M^{sp}_+(z;\eta )=M^{sp}_-(z;\eta )V^{(sp)}(z;\eta  ),\hspace{0.5cm}z \in \Sigma^{(2)},\label{jumppc}
\end{equation}
where the  jump matrix $V^{(sp)}(z,  \eta)$,  taking   $\eta=-1$ as an example,  is given by
\begin{equation}
V^{(sp)}(z, \eta=-1)=\left\{\begin{array}{llllllll}
\left(\begin{array}{cc}
1 & r(z_0) \delta^{-2} (z_0) (z-z_0)^{2i \kappa} e^{2it\theta(z)}\\
0 & 1
\end{array}\right), & z\in \Sigma_1\cup\Sigma_9,\\
\\
\left(\begin{array}{cc}
1 & 0\\
\dfrac{\bar{r}(z_0)\delta^{2} (z_0) }{1+|r(z_0)|^2}(z-z_0)^{-2i \kappa} e^{-2it\theta(z)} & 1
\end{array}\right),  &z\in \Sigma_2,\\
\\
\left(\begin{array}{cc}
1 & \dfrac{r(z_0) \delta^{-2} (z_0)}{1+|r(z_0)|^2}(z-z_0)^{2i\kappa}e^{2it\theta(z)}\\
0 & 1
\end{array}\right),  &z\in \Sigma_3,\\
\\
\left(\begin{array}{cc}
1 & 0\\
\bar{r}(z_0)\delta^{2} (z_0)(z-z_0)^{-2i \kappa} e^{-2it\theta(z)} & 1
\end{array}\right),  &z\in \Sigma_4\cup\Sigma_{12},\\
\\
\left(\begin{array}{cc}
1 & r(-z_0) \delta^{-2} (-z_0)(z+z_0)^{2i\kappa}e^{2it\theta}\\
0 & 1
\end{array}\right),  &z\in \Sigma_6\cup\Sigma_{10},\\
\\
\left(\begin{array}{cc}
1 & 0\\
\frac{\bar{r}(-z_0)\delta^{2} (-z_0)}{1+|r(-z_0)|^2}(z+z_0)^{-2i\kappa}e^{-2it\theta} & 1
\end{array}\right),  &z\in \Sigma_5,\\
\\
\left(\begin{array}{cc}
1 & \frac{r(-z_0) \delta^{-2} (-z_0)}{1+|r(-z_0)|^2}(z+z_0)^{2i\kappa}e^{2it\theta}\\
0 & 1
\end{array}\right),  &z\in \Sigma_8,\\
\\
\left(\begin{array}{cc}
1 & 0\\
\bar{r}(-z_0)\delta^{2} (-z_0)(z+z_0)^{2i\kappa}e^{-2it\theta} & 1
\end{array}\right),  &z\in \Sigma_7\cup\Sigma_{11}.\\
\end{array}\right. \label{Vsp}
\end{equation}
The  proposition  \ref{pro3v2} shows that the  jump matrix $V^{(sp)}$   uniformly  goes to  $I$   outside the neighborhood of $\pm z_0$,
 then following the result in \cite{xusp},  the  above RHP10  is solvable.  The main  contribution to the $M^{sp}(z;\eta)$  comes  from
 a local RH problem   near $\pm z_0$, see   Figure \ref{figsp}.  We simply describe the process of construction for the solution of the RHP10,
  see \cite{xusp} for the detail.  We decompose the
 jump matrix    $ V^{(sp)}=(b_-)^{-1}b_+$, and set
 $$w_\pm =\pm(b_\pm-I), \ \ \ w=w_++w_-,$$
 and let $\mu(z)$ is the solution of the operator  equation $\mu(z)=I+C_w\mu(z)$, here $C_w$ is defined by
 $$C_wf=C_+(fw_-)+C_-(fw_+),$$
 with $C_\pm$ denoting the  Cauchy projection  operators. In the same way, we can define  matrix  functions   $\mu_{\pm z_0}(z)$ and  $w_{\pm z_0}(z)$ by using
 the jumps near phase points $\pm z_0$,   then   the   solution of the RHP10  is  given by
\begin{align}
&M^{sp}(z;\eta)  =I+\frac{1}{2\pi i} \int_{\Sigma^{(2)}} \frac{\mu(s)w(s)}{s-z}ds\nonumber \\
&=I+\frac{1}{2\pi i} \int_{\Sigma^{(2)}_+} \frac{\mu_{z_0}(s)w(s)}{s-z}ds+ \frac{1}{2\pi i} \int_{\Sigma^{(2)}_- } \frac{\mu_{-z_0}(s)w(s)}{s-z}ds+\mathcal{O}(|t|^{-1}\log |t|)\nonumber\\
&=I+\frac{|t|^{-1/2}}{z-z_0}A(z_0,\eta)-\frac{|t|^{-1/2}}{ z+z_0}A(-z_0,\eta)+\mathcal{O}(|t|^{-1}\text{In}|t|), \label{asympc}
\end{align}

\begin{figure}[H]
	\centering
	\subfigure{
		\begin{tikzpicture}[node distance=2cm]
		\draw(0,0)--(2,2)node[left]{$\Sigma^{(2)}_+$};
		\draw(0,0)--(-1,1);
		\draw(0,0)--(-1,-1);
		\draw(0,0)--(2,-2);
		\draw[dashed](-6,0)--(2.2,0)node[right]{\scriptsize Re$z$};
		\draw(-3,1)--(-4,0);
		\draw(-3,-1)--(-4,0);
		\draw(-6,2)--(-4,0);
		\draw(-6,-2)--(-4,0);
		\draw[-latex](-3,1)--(-3.5,0.5);
		\draw[-latex](-3,-1)--(-3.5,-0.5);
		\draw[-latex](-6,2)--(-5,1);
		\draw[-latex](-6,-2)--(-5,-1);
		\draw[-latex](0,0)--(-0.5,-0.5);
		\draw[-latex](0,0)--(-0.5,0.5);
		\draw[-latex](0,0)--(1,1);
		\draw[-latex](0,0)--(1,-1);
		\coordinate (I) at (0,0);
		\coordinate (K) at (-4,0);
		\coordinate (l) at (-2,0);	
		\coordinate (k) at (-6,2);
		\fill (k) circle (0pt) node[left] {$\Sigma^{(2)}_-$};
		\fill (I) circle (1pt) node[below] {$z_0$};
		\fill (K) circle (1pt) node[below] {$-z_0$};
		\end{tikzpicture}
	}
	\caption{ The jump contour for the  local RHP near phase points $\pm z_0$.}
	\label{figsp}
\end{figure}
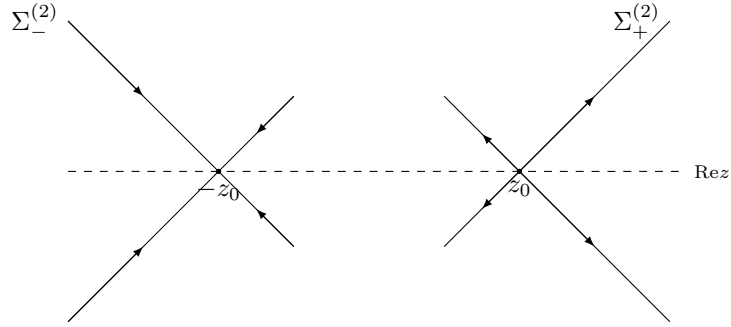

where
\begin{equation}
A(z_0,\eta)=\left(\begin{array}{cc}
0 & -A_{12}^\eta\\
A_{21}^\eta & 0
\end{array}\right), \ \ A(-z_0,\eta)= \left(\begin{array}{cc}
0 & -A_{21}^\eta \\
A_{12}^\eta  & 0
\end{array}\right),
\end{equation}
with
\begin{align}
&A_{12}^+=iz_0^{3/2}\beta_{12},\hspace{0.5cm}A_{21}^+=iz_0^{3/2}\beta_{21},\ \
A_{12}^-=-iz_0^{3/2}\beta_{21},\hspace{0.5cm}A_{21}^-=-iz_0^{3/2}\beta_{12},\\
&	\beta_{12}=\dfrac{\sqrt{2\pi}e^{-\kappa\pi/2}e^{i\pi/4}}{\bar{r}_0\Gamma(-i\kappa)},\hspace{0.5cm}\beta_{21}=\dfrac{\beta_{12} }{\kappa},  \ \ \arg(\beta_{12})=\frac{\pi}{4}+\arg\Gamma(-i\kappa)+\arg(r_0).
\end{align}
In addition, it is shown that $\parallel M^{sp}(z;\eta )\parallel_\infty\lesssim 1$.

It is easy to check that   the    RHP5 and  RHP10  have  the same contour and jump matrices,
   we  use  $M^{sp}(z;\eta )$  to  define a local model in two  circles   $z\in  U_{\pm z_0}$
\begin{equation}
M^{(\pm z_0)}(z)=M^{(out)}(z)M^{sp}(z;\eta ),\label{8.9}
\end{equation}
which  is a bounded function in $U_{\pm z_0}$ and has   the  same jump matrix as $M^{RHP}(z,\eta)$.

\section{The small norm RH problem  for error function }\label{sec9}

\quad In this section,  we consider the error matrix-function $E(z)$.
From the definition (\ref{transm4}) and (\ref{8.9}),
  we can obtain a RH problem  for the matrix function  $E(z)$.

\noindent\textbf{RHP11.}   Find a matrix-valued function $E(z)$  with following properties:

$\blacktriangleright$ Analyticity: $E(z)$ is analytical  in $\mathbb{C}\setminus (\Sigma^{(E)})$, where
$$\Sigma^{(E)}=\partial U_{z_0}\cup\partial U_{-z_0}\cup
(\Sigma^{(2)}\setminus\left( U_{z_0}\cup U_{-z_0}\right);$$

$\blacktriangleright$ Symmetry:  $\overline{E(\bar{z})}=E(-z)=\sigma_2E(z)\sigma_2^{-1}$;

$\blacktriangleright$ Asymptotic behaviours:
\begin{align}
&E(z) \sim I+\mathcal{O}(z^{-1}),\hspace{0.5cm}|z| \rightarrow \infty;
\end{align}

$\blacktriangleright$ Jump condition: $E$ has continuous boundary values $E_\pm$ on $\Sigma^{(E)}$ satisfying
$$E_+(z)=E_-(z)V^{(E)},$$
 where the jump matrix $V^{(E)}$ is given by
\begin{equation}
V^{(E)}(z)=\left\{\begin{array}{llll}
M^{(out)}(z)V^{(2)}(z)M^{(out)}(z)^{-1}, & z\in \Sigma^{(2)}\setminus U_{\pm z_0},\\[4pt]
M^{(out)}(z)M^{(sp)}(z)M^{(out)}(z)^{-1},  & z\in \partial U_{\pm z_0},
\end{array}\right. \label{deVE}
\end{equation}
 which is  shown in  Figure \ref{figE}.

We will show  that for large times, the error function  $E(z)$  solves following small norm  RH problem.

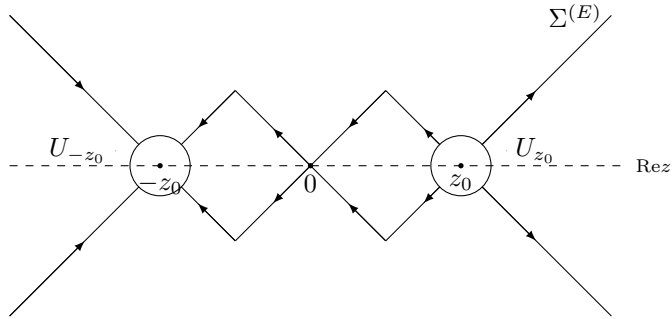
\begin{figure}[H]
	\centering
	\subfigure{
		\begin{tikzpicture}[node distance=2cm]
		\draw[ ](0.28,0.28)--(2,2)node[left]{$\Sigma^{(E)}$};
		\draw[](-0.28,0.28)--(-1,1);
		\draw[](-0.28,-0.28)--(-1,-1);
		\draw[](0.28,-0.28)--(2,-2);
		\draw[dashed](-6,0)--(2.2,0)node[right]{\scriptsize Re$z$};
		\draw[](-1,1)--(-2,0);
		\draw[](-1,-1)--(-2,0);
		\draw[](-2,0)--(-3,1);
		\draw[](-2,0)--(-3,-1);
		\draw[](-3,1)--(-3.72,0.28);
		\draw[](-3,-1)--(-3.72,-0.28);
		\draw[](-6,2)--(-4.28,0.28);
		\draw[](-6,-2)--(-4.28,-0.28);
		\draw[][-latex](-1,1)--(-1.5,0.5);
		\draw[][-latex](-1,-1)--(-1.5,-0.5);
		\draw[][-latex](-2,0)--(-2.5,-0.5);
		\draw[][-latex](-2,0)--(-2.5,0.5);
		\draw[][-latex](-3,1)--(-3.5,0.5);
		\draw[][-latex](-3,-1)--(-3.5,-0.5);
		\draw[][-latex](-6,2)--(-5,1);
		\draw[][-latex](-6,-2)--(-5,-1);
		\draw[][-latex](-0.28,-0.28)--(-0.5,-0.5);
		\draw[][-latex](-0.28,0.28)--(-0.5,0.5);
		\draw[][-latex](0.28,0.28)--(1,1);
		\draw[][-latex](0.28,-0.28)--(1,-1);
		\draw[](0,0) circle (0.4);
		\draw[](-4,0) circle (0.4);
		\coordinate (A) at (0.6,0.2);
		\fill (A) circle (0pt) node[right][] {$U_{z_0}$};
		\coordinate (I) at (0,0);
		\coordinate (B) at (-4.6,0.2);
		\fill (B) circle (0pt) node[left][] {$U_{-z_0}$};
		\coordinate (K) at (-4,0);
		\coordinate (l) at (-2,0);	
		\coordinate (k) at (-1,-0.1);
		\fill (I) circle (1pt) node[below] {$z_0$};
		\fill (K) circle (1pt) node[below] {$-z_0$};
		\fill (l) circle (1pt) node[below] {$0$};
		\end{tikzpicture}
	}
	\caption{  The jump contour $\Sigma^{(E)}$ for the $E(z)$ .}
	\label{figE}
\end{figure}

By using (\ref{normmout}) and   {Proposition \ref{pro3v2}}, we have the following estimates
\begin{equation}
|V^{(E)}-I|\lesssim\left\{\begin{array}{llll}
\exp\left\{-t\frac{\sqrt{2}}{32z_0^2}|z\mp z_0|\right\},  & z\in \Sigma^{(2)}_\pm\setminus U_{\pm z_0},\\[6pt]
\exp\left\{-\frac{|t|}{4z_0}\right\},   & z\in \Sigma^{(2)}_0.
\end{array}\right. \label{VE-I}
\end{equation}

For $z\in \partial U_{\pm z_0}$,  $M^{(out)}(z)$ is bounded,  so   by using  (\ref{asympc}),  we find that
\begin{equation}
| V^{(E)}-I|=   \big|M^{(out)}(z)^{-1}(M^{sp}(z)-I)M^{(out)}(z) \big| = \mathcal{O}(|t|^{-1/2}).\label{VE}
\end{equation}
Therefore,    the   existence and uniqueness  of  the RHP11 can  shown  by using  a  small-norm RH problem \cite{RN5,RN9,RN10},  and  we have
\begin{equation}
E(z)=I+\frac{1}{2\pi i}\int_{\Sigma^{(E)}}\dfrac{\left( I+\rho(s)\right) (V^{(E)}-I)}{s-z}ds,\label{Ez}
\end{equation}
where the $\rho\in L^2(\Sigma^{(E)})$ is the unique solution of following equation:
\begin{equation}
(1-C_E)\rho=C_E\left(I \right),
\end{equation}
where $C_E$ is a integral operator defined by
\begin{equation}
C_E(f)(z)=C_-\left( f(V^{(E)}-I)\right) ,
\end{equation}
where the $C_-$ is the usual Cauchy projection operator on $\Sigma^{(E)}$
\begin{equation}
C_-(f)(s)=\lim_{z\to \Sigma^{(E)}_-}\frac{1}{2\pi i}\int_{\Sigma^{(E)}}\dfrac{f(s)}{s-z}ds.
\end{equation}
Then by (\ref{VE}) we have
\begin{equation}
\parallel C_E\parallel\leq\parallel C_-\parallel \parallel V^{(E)}-I\parallel_\infty \lesssim \mathcal{O}(t^{-1/2}),
\end{equation}
which means $\parallel C_E\parallel<1$ for sufficiently large t,   therefore  $1-C_E$ is invertible,  and   $\rho$  exists and is unique.
Moreover,
\begin{equation}
\parallel \rho\parallel_{L^2(\Sigma^{(E)})}\lesssim\dfrac{\parallel C_E\parallel}{1-\parallel C_E\parallel}\lesssim|t|^{-1/2}.\label{normrho}
\end{equation}
Then we have the existence and boundedness of $E(z)$. In order to reconstruct the solution $u(y,t)$ of (\ref{sp}), we need the asymptotic behavior of $E(z)$ as $z\to 0$ and the long time asymptotic behavior of $E(0)$. Note that when we estimate its  asymptotic behavior, from (\ref{Ez}) and (\ref{VE-I}) we only need to consider the calculation on $\partial U_{\pm z_0}$ because it  approach zero exponentially on other boundary.
\begin{proposition}\label{asyE}
	As $z\to 0$, we have
	\begin{align}
	E(z)=E(0)+E_1z+\mathcal{O}(z^2),
	\end{align}
	where
	\begin{align}
	E(0)=I+\frac{1}{2\pi i}\int_{\Sigma^{(E)}}\dfrac{\left( I+\rho(s)\right) (V^{(E)}-I)}{s}ds,
	\end{align}
	with long time asymptotic behavior
	\begin{equation}
	E(0)=I+|t|^{-1/2}H^{(0)}+\mathcal{O}(|t|^{-1}),\label{E0t}
	\end{equation}
	and
	\begin{align}
	H^{(0)}&=\frac{1}{2\pi i}\int_{\partial U_{\pm z_0}}\dfrac{M^{(out)}(s)^{-1}A(\pm z_0,\eta)M^{(out)}(s)}{s(\pm s- z_0)}ds\nonumber\\
	&=\frac{1}{z_0} M^{(out)}(z_0)^{-1}A( z_0,\eta)M^{(out)}(z_0)+\nonumber\\
	&+\frac{1}{z_0}M^{(out)}(-z_0)^{-1}A(- z_0,\eta)M^{(out)}(-z_0) .
	\end{align}
	The last equality  follows from a residue calculation. Moreover,
	\begin{equation}
	E(0)^{-1}=I+\mathcal{O}(|t|^{-1/2}).
	\end{equation}
	And
	\begin{equation}
	E_1=-\frac{1}{2\pi i}\int_{\Sigma^{(E)}}\dfrac{\left( I+\rho(s)\right) (V^{(E)}-I)}{s^2}ds,
	\end{equation}
	satisfying long time asymptotic behavior condition
	\begin{equation}
	E_1=|t|^{-1/2}H^{(1)}+\mathcal{O}(|t|^{-1}),\label{E1t}
	\end{equation}
	where
	\begin{align}
	H^{(1)}=&\frac{1}{2\pi i}\int_{\partial U_{\pm z_0}}\dfrac{M^{(out)}(s)^{-1}A(\pm z_0,\eta)M^{(out)}(s)}{s^2(\pm s- z_0)}ds\nonumber\\
	=&\dfrac{1}{z_0^2}M^{(out)}(z_0)^{-1}A(z_0,\eta)M^{(out)}(z_0)\nonumber\\
	&-\dfrac{1}{z_0^2}[M^{(out)}(-z_0)^{-1}A(-z_0,\eta)M^{(out)}(-z_0)  .
	\end{align}
\end{proposition}

\section{Analysis of   the pure $\bar{\partial}$-Problem}
\quad Now we consider the proposition and the long time asymptotics behavior of $M^{(3)}$.
The RHP6  of $M^{(3)}$ is equivalent to the integral equation
\begin{equation}
M^{(3)}(z)=I+\frac{1}{\pi}\int_\mathbb{C}\dfrac{\bar{\partial}M^{(3)}(s)}{z-s}dm(s)=I+\frac{1}{\pi}\int_\mathbb{C}\dfrac{M^{(3)}(s)W^{(3)}(s)}{z-s}dm(s),
\end{equation}
where $m(s)$ is the Lebegue measure on the $\mathbb{C}$. If we denote $C_z$ is the left Cauchy-Green integral  operator,
\begin{equation*}
fC_z(z)=\frac{1}{\pi}\int_C\dfrac{f(s)W^{(3)}(s)}{z-s}dm(s),
\end{equation*}
then above equation can be rewritten as
\begin{equation}
M^{(3)}(z)=I\left(I-C_z \right) ^{-1}.\label{deM3}
\end{equation}
To proof the existence of operator $\left(I-C_z \right) ^{-1}$, we have following Lemma.
\begin{lemma}\label{Cz}
	The norm of the integral operator $C_z$ decay to zero as $t\to\infty$:
	\begin{equation}
	\parallel C_z\parallel_{L^\infty\to L^\infty}\lesssim |t|^{-1/6},
	\end{equation}
	which implies that  $\left(I-C_z \right) ^{-1}$ exists.
\end{lemma}
\begin{proof}
	For any $f\in L^\infty$,
	\begin{align}
	\parallel fC_z \parallel_{L^\infty}&\leq \parallel f \parallel_{L^\infty}\frac{1}{\pi}\int_C\dfrac{|W^{(3)}(s)|}{|z-s|}dm(s)\nonumber\\
	&\lesssim\parallel f \parallel_{L^\infty}\frac{1}{\pi}\int_C\dfrac{|\bar{\partial}R^{(2)}(s)|}{|z-s|}dm(s).
	\end{align}
	So we only need to  estimate the integral
	\begin{equation*}
	\frac{1}{\pi}\int_C\dfrac{|\bar{\partial}R^{(2)} (s)|}{|z-s|}dm(s).
	\end{equation*}
	We only show  the case $\eta=-1$. For $\bar{\partial}R^{(2)} (s)$ is a piece-wise function,
we  prove  the case  in the region $\Omega_1$, the other  regions are similar.  By using (\ref{dbarRj}),  we have
	\begin{align}
	\int_{\Omega_1}\dfrac{|\bar{\partial}R^{(2)} (s)|}{|z-s|}dm(s)\leq F_1+F_2+F_3,
	\end{align}
	where
	\begin{align}
	&F_1=\int^{+\infty}_{0}\int^{+\infty}_{z_0+v}\dfrac{|\bar{\partial}X_Z (s)|e^{-\frac{vt}{2(u^2+v^2)}}}{\sqrt{(u-x)^2+(v-y)^2}}due^{-2tv\xi}dv;\\
	&F_2=\int^{+\infty}_{0}\int^{+\infty}_{z_0+v}\dfrac{|p_1'(u) |e^{-\frac{vt}{2(u^2+v^2)}}}{\sqrt{(u-x)^2+(v-y)^2}}due^{-2tv\xi}dv;\\
	&F_3=\int^{+\infty}_{0}\int^{+\infty}_{z_0+v}\dfrac{\left( (u-z_0)^2+v^2\right)^{-1/4} e^{-\frac{vt}{2(u^2+v^2)}}}{\sqrt{(u-x)^2+(v-y)^2}}due^{-2tv\xi}dv.
	\end{align}
and we denote $s=u+vi$, $z=x+yi$.
	
In the following calculation,  we will use the inequality
	\begin{equation}
	\parallel |s-z|^{-1}\parallel_{L^2(z_0, +\infty)}^2=\int^{+\infty}_{z_0}\frac{1}{|v-y|}\left[  \left( \frac{u-x}{v-y}\right) ^2+1\right]  ^{-1} d\left( \frac{u-x}{|v-y|}\right) \leq\frac{\pi}{|v-y|}.
	\end{equation}
	To deal with the absolute value sign, we suppose $y>0$. In fact $y<0$ we can directly remove the absolute value sign and use the same way to estimates it.
	
For $F_1$, noting that $-\frac{vt}{2(u^2+v^2)}$ is a monotonic  decreasing function of $u$,  so
	\begin{align}
	F_1 &\leq\int^{+\infty}_{0}\parallel |s-z|^{-1}\parallel_{L^2(z_0, +\infty)}\parallel\bar{\partial}X_Z (s)\parallel_{L^2(z_0, +\infty)}e^{-\frac{vt}{2(z_0^2+v^2)}}e^{-2tv\xi}dv\nonumber\\
	&\lesssim\int^{+\infty}_{0} |v-y|^{-1/2}\exp\left( -\frac{v|t|}{2}\left( \frac{1}{z_0^2}-\frac{1}{z_0^2+v^2}\right) \right)dv\nonumber\\
	&=  \int_{0}^{y}(y-v)^{-1/2}\exp\left( -\frac{v|t|}{2}\left( \frac{1}{z_0^2}-\frac{1}{z_0^2+v^2}\right)\right) dv\nonumber\\
	&+ \int_{y}^{+\infty}(v-y)^{-1/2}\exp\left( -\frac{v|t|y^2}{2z_0^2(y^2+z_0^2)}\right)dv.\label{10.10}
	\end{align}

	For the first item, note that $e^{-z}\leq z^{1/6}$ for all $z>0$, then
	\begin{align}
	&\int_{0}^{y}(y-v)^{-1/2}\exp\left( -\frac{v|t|}{2}\left( \frac{1}{z_0^2}-\frac{1}{z_0^2+v^2}\right)\right) dv \nonumber\\
	&\lesssim\int_{0}^{y}(y-v)^{-1/2}v^{-1/2} dv|t|^{-1/4}\lesssim |t|^{-1/6}.\label{1011}
	\end{align}

	For the last integral we make the substitution $w=v-y$  then we get
	\begin{align}
	&\int_{y}^{+\infty}(v-y)^{-1/2}\exp\left( -\frac{v|t|y^2}{2z_0^2(y^2+z_0^2)}\right)dv\nonumber\\
	&\leq\int_{0}^{+\infty}w^{-1/2}\exp\left( -\frac{w|t|y^2}{2z_0^2(y^2+z_0^2)}\right)dw\exp\left( -\frac{|t|y^3}{2z_0^2(y^2+z_0^2)}\right)\lesssim |t|^{-1/2}.\label{10.12}
	\end{align}
	Substituting (\ref{10.12}) and (\ref{1011}) into (\ref{10.10}) gives
	\begin{align}
	F_1 \lesssim |t|^{-1/4}.\label{10.13}
	\end{align}
	
The  $F_2$ has the same estimate with (\ref{10.13}).    And for $F_3$,  we first  have that
	\begin{align}
	&\parallel \left( (u-z_0)^2+v^2\right)^{-1/4}\parallel_{L^p(z_0,+\infty)}\nonumber\\
	&=\left\lbrace \int_{z_0}^{+\infty}\left[ (u-z_0)^2+v^2\right]^{-p/4}dv\right\rbrace  ^{1/p}\nonumber\\
	&=\left\lbrace \int_{z_0}^{+\infty}\left[ 1+\left( \frac{u-z_0}{v}\right) ^2\right]^{-p/4}d\left( \frac{u-z_0}{v}\right)  \right\rbrace ^{1/p}v^{1/p-1/2}
 \lesssim v^{1/p-1/2},\label{F3}
	\end{align}
	and
	\begin{align}
	\parallel |s-z|^{-1}\parallel_{L^q(z_0,+\infty)}&=\left\lbrace \int_{z_0}^{+\infty}\left[  \left( \frac{u-x}{v-y}\right) ^2+1\right]^{-q/2}
 d\left( \frac{u-x}{|v-y|}\right)\right\rbrace ^{1/q}|v-y|^{1/q-1}\nonumber\\
	&\lesssim |v-y|^{1/q-1},\label{nromu}
	\end{align}
	where $p>2$ and $\frac{1}{p}+\frac{1}{q}=1$.
 Then we have
	\begin{align}
	F_3&\leq\int^{+\infty}_{0} 	\parallel |s-z|^{-1}\parallel_{L^q}\parallel \left( (u-z_0)^2+v^2\right)^{-1/4}\parallel_{L^p}\exp\left( -\frac{v|t|}{2}\left( \frac{1}{z_0^2}-\frac{1}{z_0^2+v^2}\right)\right)dv\nonumber\\
	&\lesssim \int^{+\infty}_{0}v^{1/p-1/2} |v-y|^{1/q-1}\exp\left( -\frac{v|t|}{2}\left( \frac{1}{z_0^2}-\frac{1}{z_0^2+v^2}\right)\right)dv\nonumber\\	
	&\lesssim\int^{y}_{0}v^{1/p-1/2} (y-v)^{1/q-1}\exp\left( -\frac{v|t|}{2}\left( \frac{1}{z_0^2}-\frac{1}{z_0^2+v^2}\right)\right)dv\nonumber\\
	&+\int^{+\infty}_{y}v^{1/p-1/2} (v-y)^{1/q-1}\exp\left( -\frac{v|t|}{2}\left( \frac{1}{z_0^2}-\frac{1}{z_0^2+v^2}\right)\right)dv.
	\end{align}
	For the first term,  using  the inequality $e^{-z}\leq z^{-1/6}$  leads to
	 \begin{align}
	 &\int^{y}_{0}v^{1/p-1/2} (y-v)^{1/q-1}\exp\left( -\frac{v|t|}{2}\left( \frac{1}{z_0^2}-\frac{1}{z_0^2+v^2}\right)\right)dv\nonumber\\
	&\lesssim |t|^{-1/6}\int^{y}_{0}v^{1/p-1} (y-v)^{1/q-1}dv
	 \lesssim |t|^{-1/6}.
	 \end{align}
	 And for the second term,   we estimate similarly as we estimate $F_1$. Let $w=v-y$, then we have
	 \begin{align}
	 &\int^{+\infty}_{y}v^{1/p-1/2} (v-y)^{1/q-1}\exp\left( -\frac{v|t|}{2}\left( \frac{1}{z_0^2}-\frac{1}{z_0^2+v^2}\right)\right)dv\nonumber\\
	&\leq\int^{+\infty}_{0} w^{1/q-1}(w+y)^{1/p-1/2}\exp\left( -\frac{w|t|y^2}{2z_0^2(y^2+z_0^2)}\right)dw\exp\left( -\frac{|t|y^3}{2z_0^2(y^2+z_0^2)}\right)\nonumber\\
	&\lesssim\int^{+\infty}_{0} w^{-1/2}\exp\left( -\frac{w|t|y^2}{2z_0^2(y^2+z_0^2)}\right)dw \lesssim |t|^{-1/2}.\label{F32}
	 \end{align}
	 Finally,  we have
	 \begin{equation}
	 F_3\lesssim |t|^{-1/6}.
	 \end{equation}
	
Summary the results obtained  above, we obtain the finally consequence.
\end{proof}

\begin{proposition}\label{asyM30}
The solution $M^{(3)}(z)$  of  the RHP6 admits the following estimate
\begin{align}
\parallel M^{(3)}(0)-I\parallel\lesssim |t|^{-1}.
\end{align}
\end{proposition}
\begin{proof}
	We only prove  the case $\eta=-1$. Because of the boundedness of $M^{RHP}$, we only need to  estimate
	\begin{equation*}
	\frac{1}{\pi}\int_C\dfrac{|\bar{\partial}R^{(2)} (s)|}{|z-s|}dm(s).
	\end{equation*}
	For $\bar{\partial}R^{(2)} (s)$ is a piece-wise function, we show  the case
 in the region $\Omega_1$,  the other regions are similar. Denote $s=u+vi$, then we have
	\begin{align}
	\int_{\Omega_1}\dfrac{|\bar{\partial}R^{(2)} (s)|}{|s|}dm(s)=\int^{+\infty}_{0}\int^{+\infty}_{z_0+v}\dfrac{|\bar{\partial}R_1 (s)|e^{-\frac{vt}{2(u^2+v^2)}}}{(u^2+v^2)^{1/2}}e^{-2tv\xi}dudv\leq I_1+I_2+I_3,
	\end{align}
	where by using  (\ref{dbarRj}), we have
	\begin{align}
	&I_1=\int^{+\infty}_{0}\int^{+\infty}_{z_0+v}\dfrac{|\bar{\partial}X_Z (s)|e^{-\frac{vt}{2(u^2+v^2)}}}{(u^2+v^2)^{1/2}}e^{-2tv\xi}dudv;\\
	&I_2=\int^{+\infty}_{0}\int^{+\infty}_{z_0+v}\dfrac{|p_1'(u) (s)|e^{-\frac{vt}{2(u^2+v^2)}}}{(u^2+v^2)^{1/2}}e^{-2tv\xi}dudv;\\
	&I_3=\int^{+\infty}_{0}\int^{+\infty}_{z_0+v}\dfrac{\left( (u-z_0)^2+v^2\right)^{-1/4} e^{-\frac{vt}{2(u^2+v^2)}}}{(u^2+v^2)^{1/2}}e^{-2tv\xi}dudv.
	\end{align}
	
For  $I_1$, we divide it to two part:
	\begin{equation}
	I_1=\int^{+\infty}_{z_0}\int^{+\infty}_{z_0+v}\dfrac{|\bar{\partial}X_Z (s)|e^{-\frac{vt}{2(u^2+v^2)}}}{(u^2+v^2)^{1/2}}e^{-2tv\xi}dudv+\int^{z_0}_{0}\int^{+\infty}_{z_0+v}\dfrac{|\bar{\partial}X_Z (s)|e^{-\frac{vt}{2(u^2+v^2)}}}{(u^2+v^2)^{1/2}}e^{-2tv\xi}dudv.
	\end{equation}
	For the first integral, we use (\ref{nromu}) in the case of $y=0$, then we have
	\begin{align}
	&\int^{+\infty}_{z_0}\int^{+\infty}_{z_0+v}\dfrac{|\bar{\partial}X_\mathcal{Z} (s)|e^{-\frac{vt}{2(u^2+v^2)}}}{(u^2+v^2)^{1/2}}e^{-2tv\xi}dudv\nonumber\\
	&\leq\int^{+\infty}_{z_0}\parallel |s-z|^{-1}\parallel_{L^2}\parallel\bar{\partial}X_Z (s)\parallel_{L^2}e^{-\frac{vt}{2((v+z_0)^2+v^2)}}e^{-2tv\xi}dv\nonumber\\
	&\lesssim \int^{+\infty}_{z_0} v^{-1/2}\exp\left\lbrace-\frac{|2t|v}{5z_0^2} \right\rbrace dv\nonumber\\
	&\lesssim \int^{+\infty}_{z_0} \exp\left\lbrace-\frac{|2t|v}{5z_0^2} \right\rbrace dv \lesssim |t|^{-1}.
	\end{align}
	And for the second item,
	\begin{align}
	&\int^{z_0}_{0}\int^{+\infty}_{z_0+v}\dfrac{|\bar{\partial}X_\mathcal{Z} (s)|e^{-\frac{vt}{2(u^2+v^2)}}}{(u^2+v^2)^{1/2}}e^{-2tv\xi}dudv\nonumber\\
	&\leq\int^{z_0}_{0}\left((v+z_0)^2+v^2 \right) ^{-1/2}\exp\left(\frac{|t|v}{(v+z_0)^2+v^2} \right)dv \nonumber\\
	&\lesssim\int^{z_0}_{0}\exp\left(\frac{|t|v}{z_0^2} \right)dv \lesssim |t|^{-1}.
	\end{align}
	And using the same way we can estimate $I_2$ and get same result. Finally we bound $I_3$, similarly we divide it to two parts:
	\begin{align}
	I_3=&\int^{+\infty}_{z_0}\int^{+\infty}_{z_0+v}\dfrac{\left( (u-z_0)^2+v^2\right)^{-1/4} e^{-\frac{vt}{2(u^2+v^2)}}}{(u^2+v^2)^{1/2}}e^{-2tv\xi}dudv\nonumber\\
	&+\int^{z_0}_{0}\int^{+\infty}_{z_0+v}\dfrac{\left( (u-z_0)^2+v^2\right)^{-1/4} e^{-\frac{vt}{2(u^2+v^2)}}}{(u^2+v^2)^{1/2}}e^{-2tv\xi}dudv.
	\end{align}
	For the first integral, we use (\ref{nromu}) and (\ref{F3}) similarly, take $p>2$ and $\frac{1}{p}+\frac{1}{q}=1$:
	\begin{align}
	&\int^{+\infty}_{z_0}\int^{+\infty}_{z_0+v}\dfrac{\left( (u-z_0)^2+v^2\right)^{-1/4} e^{-\frac{vt}{2(u^2+v^2)}}}{(u^2+v^2)^{1/2}}e^{-2tv\xi}dudv\nonumber\\
	&\leq\int^{+\infty}_{z_0} 	\parallel |s|^{-1}\parallel_{L^q}\parallel \left( (u-z_0)^2+v^2\right)^{-1/4}\parallel_{L^p}\exp\left( -\frac{v|t|}{2}\left( \frac{1}{z_0^2}-\frac{1}{z_0^2+v^2}\right)\right)dv\nonumber\\
	&\lesssim \int^{+\infty}_{z_0} v^{-1/2}\exp\left\lbrace-\frac{|2t|v}{5z_0^2} \right\rbrace dv\nonumber\\
	&\lesssim \int^{+\infty}_{z_0} \exp\left\lbrace-\frac{|2t|v}{5z_0^2} \right\rbrace dv \lesssim |t|^{-1}.
	\end{align}
	And for the second integral,  note that
	\begin{align}
	&\parallel (u^2+v^2)^{-1/2}\exp\left\lbrace\frac{|t|v}{2(u^2+v^2)} \right\rbrace \parallel_{L^4(z_0+v, +\infty)}\nonumber\\
	&=\left\lbrace \int^{+\infty}_{z_0} (u^2+v^2)^{-2}\exp\left\lbrace\frac{2|t|v}{u^2+v^2}\right\rbrace du \right\rbrace^{1/4} \nonumber\\
	&=\left\lbrace \int^{+\infty}_{z_0} |8tv|^{-1}u^{-1}\exp\left\lbrace\frac{2|t|v}{u^2+v^2}\right\rbrace' du \right\rbrace^{1/4} \nonumber\\
	&\lesssim |t|^{-1/4}v^{-1/4}\left( \exp\left\lbrace\frac{|t|v}{2((v+z_0)^2+v^2)}\right\rbrace+1\right) .
	\end{align}
	Then together with (\ref{F3}) we obtain
	\begin{align}
	&\int^{z_0}_{0}\int^{+\infty}_{z_0+v}\dfrac{\left( (u-z_0)^2+v^2\right)^{-1/4} e^{-\frac{vt}{2(u^2+v^2)}}}{(u^2+v^2)^{1/2}}e^{-2tv\xi}dudv\nonumber\\
	&\leq\int^{z_0}_{0} \parallel \left( (u-z_0)^2+v^2\right)^{-1/4}\parallel_{L^{4/3}}\parallel (u^2+v^2)^{-1/2}\exp\left\lbrace\frac{|t|v}{2(u^2+v^2)} \right\rbrace \parallel_{L^4}e^{-2tv\xi}dv\nonumber\\
	&\lesssim |t|^{-1/4}\int^{z_0}_{0}v^{3/4-1/2}v^{-1/4}e^{-2tv\xi}\left( \exp\left\lbrace\frac{|t|v}{2((v+z_0)^2+v^2)}\right\rbrace+1\right)dv\nonumber\\
	&\lesssim |t|^{-1/4}\int^{z_0}_{0}e^{-2tv\xi}dv
	\lesssim |t|^{-5/4}.\label{I3}
	\end{align}
	So we have
	\begin{equation}
	I_3\lesssim |t|^{-1}.
	\end{equation}
	We come to the  result by combining above equations.	
\end{proof}

To reconstruct the solution $u(y,t)$ of the SP equation  (\ref{sp}), we need   the asymptotic behavior of $M^{(3)}_1$   given by
\begin{equation}
M^{(3)}(z)=M^{(3)}(0)+M^{(3)}_1(x,t)z+\mathcal{O}(z^2),\ \ z\to0\label{expM3}
\end{equation}
and
\begin{equation}
M^{(3)}_1(y,t)=\frac{1}{\pi}\int_C\frac{M^{(3)}(s)W^{(3)}(s)}{s^2}dm(s).
\end{equation}
The  $M^{(3)}_1$ admits the following  estimate.
\begin{lemma}
	For all $t\neq0$, we have
	\begin{equation}
	|M^{(3)}_1(x,t)|\lesssim |t|^{-1}.\label{M31}
	\end{equation}
\end{lemma}
\begin{proof}
	From  {Lemma \ref{Cz}} and (\ref{deM3}), we have $\parallel M^{(3)}\parallel_\infty \lesssim1$. And we only estimate the integral on $\Omega_1$ since the other estimates are similar.
Like in the above Lemma, by (\ref{dbarRj}) and (\ref{DBARR2}) we obtain
	\begin{equation}
	|\frac{1}{\pi}\int_{\Omega_1}M^{(3)}(s)\bar{\partial}R^{(2)}(s)|s|^{-2}dm(s)|\lesssim \frac{1}{\pi}\int_{\Omega_1}|\bar{\partial}R^{(2)}(s)||s|^{-2}dm(s)\lesssim I_4+I_5+I_6,
	\end{equation}	
	where the last inequality is from (\ref{dbarRj}) and we also have for $s=u+vi$,
	\begin{align}
	&I_4=\int^{+\infty}_{0}\int^{+\infty}_{z_0+v}\dfrac{|\bar{\partial}X_Z (s)|e^{-\frac{vt}{2(u^2+v^2)}}}{u^2+v^2}dve^{-2tv\xi}du;\\
	&I_5=\int^{+\infty}_{0}\int^{+\infty}_{z_0+v}\dfrac{|p_1'(u) (s)|e^{-\frac{vt}{2(u^2+v^2)}}}{u^2+v^2}dve^{-2tv\xi}du;\\
	&I_6=\int^{+\infty}_{0}\int^{+\infty}_{z_0+v}\dfrac{\left( (u-z_0)^2+v^2\right)^{-1/4} e^{-\frac{vt}{2(u^2+v^2)}}}{u^2+v^2}dve^{-2tv\xi}du.
	\end{align}
	Note that for all $s$ in $\Omega_1$
	\begin{equation*}
	(u^2+v^2)^{-1/2}\leq \frac{1}{z_0},
	\end{equation*}
	so we have
	\begin{equation}
	I_j\leq\frac{1}{z_0}I_{j-3},\hspace{0.4cm} \text{for } j=4,5,6.
	\end{equation}
	So from Proposition \ref{asyM30},  we can easily get the result.
\end{proof}

\section{Soliton resolution  for the  SP equation }

\quad Now we begin to construct the long time asymptotics of the SP equation (\ref{sp}).
 Inverting the sequence of transformations (\ref{transm1}), (\ref{transm2}), (\ref{transm3}) and (\ref{transm4}), we have
\begin{align}
M(z)=&M^{(3)}(z)E(z)M^{(out)}(z)R^{(2)}(z)^{-1}T(z)^{-\sigma_3},\hspace{0.5cm}  z \in C\setminus U_{\pm z_0}
\end{align}
To  reconstruct the solution $u(y,t)$ by using (\ref{recons u}),   we take $z\to0$ along the  imaginary axis. In this case,  $ R^{(2)}(z)=I$, and  we have
\begin{align}
u(x,t)&=u(y(x,t)),t) = -i \lim_{z\to 0} z \left(   M(0)^{-1} M (z) \right)_{12}.
\end{align}
Further using   {Propositions} \ref{proT}, \ref{asyE}  and  \ref{asyM30},  we can obtain the long time asymptotics behavior
\begin{align}
&(M(0)^{-1} M (z) = M^{(out)}(0)^{-1}M^{(out)}(z)+ M^{(out)}(0)^{-1}M^{(out)}(z)T_1^{-\sigma_3}z\nonumber\\
&+M^{(out)}(0)^{-1}H^{(1)} M^{(out)}(z)z|t|^{-1/2}+\mathcal{O}(|t|^{-1}),
\end{align}
where   $T(z)^{\sigma_3}$  is a  diagonal matrix,  then by  {corollary} \ref{usol} and simply calculation  we finally obtain following result.

\begin{theorem}\label{last}   Let $q(x,t)$ be the solution for  the initial-value problem (\ref{sp})-(\ref{sp2}) with generic data   $u_0(x)\in H^{1,1}(\mathbb{R})$,
For  fixed  $y_1,y_2,v_1,v_2\in \mathbb{R}$ with $y_1\leq y_2$  and $v_1\leq v_2\in \mathbb{R}^-$,  we define two zones for spectral variable $z$ 
  \begin{align}
 & I=\left\lbrace z:\ -1/(4v_1)<|z|^2<-1/(4v_2)\}, \ \ N(I)= \left\lbrace z_k\in \mathcal{Z}:\ z_k\in I\right\rbrace  \right\rbrace
  \end{align}
 and a cone for variables $y, t$
	\begin{equation*}
	C(y_1,y_2,v_1,v_2)=\left\lbrace (y,t)\in R^2|y=y_0+vt\text{ ,with }y_0\in[y_1,y_2]\text{, }v\in[v_1,v_2]\right\rbrace,
	\end{equation*}
	which are shown in {Figure \ref{figC(I)}}.  Denote $u_{sol}(y,t|D(I))$ be the $N(I)$ soliton solution corresponding to   scattering data
$ \left\lbrace z_k,c_k(I)\right\rbrace_{k=1}^{N(I)}  $ which given in (\ref{dataI}) and  corresponding $c_+(x,t|D(I))$ defined by (\ref{c+D}). 
Then as $|t|\to\infty$ with $(y,t)\in C(y_1,y_2,v_1,v_2)$,  we have
	\begin{align}
	u(x,t)=&u(y(x,t),t)=u_{sol}(y(x,t),t|D(I)) -i|t|^{-1/2} f_{12}(y,t) +\mathcal{O}(|t|^{-1}),\label{resultu}
	\end{align}
where 
	\begin{align}
&	y(x,t)=x-c_+(x,t;D(I))-iT_1^{-1}
	 -i|t|^{-1/2}f_{11}(y,t) +\mathcal{O}(|t|^{-1})\label{resulty},\\
&f_{12}(y(x,t),t)=\left[M^{(out)}(0)^{-1}H^{(1)}M^{(out)}(0) \right] _{12},\nonumber\\
	&f_{11}(y(x,t),t)=\left[M^{(out)}(0)^{-1}H^{(1)}M^{(out)}(0) \right] _{11}.\nonumber
	\end{align}

\end{theorem}
The   long time asymptotic expansion  (\ref{resultu}) shows the soliton resolution  of  for  the initial value problem  of  the short-pluse equation,
which  consisting of three terms:  the leading order term can
be characterized with  an $N(I)$-soliton whose parameters are modulated by
a sum of localized soliton-soliton
 interactions as one moves through the cone; the second $t^{-1/2}$ order term  coming from soliton-radiation interactions on continuous   spectrum  up to an residual error order
 $\mathcal{O}(|t|^{-1})$ from a $\overline\partial$ equation. Our results also show that soliton solutions of
  short-pluse  equation are asymptotically stable.\vspace{6mm}

\noindent\textbf{Acknowledgements}

This work is supported by  the National Science
Foundation of China (Grant No. 11671095,  51879045).

\hspace*{\parindent}
\\

\end{document}